\documentclass[lettersize,journal]{IEEEtran}
\usepackage{amsmath,amsfonts, amssymb, amsthm}
\usepackage{algorithmic}
\usepackage{array}
\usepackage[caption=false,font=normalsize,labelfont=sf,textfont=sf]{subfig}
\usepackage{textcomp}
\usepackage{stfloats}
\usepackage{url}
\usepackage{verbatim}
\usepackage{graphicx}
\usepackage{hyperref}

\newtheorem{definition}{Definition}
\newtheorem{theorem}{Theorem}
\newtheorem{proposition}{Proposition}
\newtheorem{corollary}{Corollary}

\def\BibTeX{{\rm B\kern-.05em{\sc i\kern-.025em b}\kern-.08em
    T\kern-.1667em\lower.7ex\hbox{E}\kern-.125emX}}
\usepackage{balance}
\begin{document}
\title{Frames and vertex-frequency representations in graph fractional Fourier domain}

\author{Linbo~Shang and Zhichao~Zhang,~\IEEEmembership{Member,~IEEE}
	\thanks{This work was supported in part by the Foundation of Key Laboratory of System Control and Information Processing, Ministry of Education under Grant Scip20240121; in part by the Foundation of Key Laboratory of Computational Science and Application of Hainan Province under Grant JSKX202401; and in part by the Foundation of Key Laboratory of Numerical Simulation of Sichuan Provincial Universities under Grant KLNS--2024SZFZ005. \emph{(Corresponding author: Zhichao~Zhang.)}}
	\thanks{Linbo~Shang is with the School of Mathematics and Statistics, Nanjing University of Information Science and Technology, Nanjing 210044, China (e-mail: ashanglinbo@163.com).}
	\thanks{Zhichao~Zhang is with the School of Mathematics and Statistics, Nanjing University of Information Science and Technology, Nanjing 210044, China, with the Key Laboratory of System Control and Information Processing, Ministry of Education, Shanghai Jiao Tong University, Shanghai 200240, China, with the Key Laboratory of Computational Science and Application of Hainan Province, Hainan Normal University, Haikou 571158, China, and also with the Key Laboratory of Numerical Simulation of Sichuan Provincial Universities, School of Mathematics and Information Sciences, Neijiang Normal University, Neijiang 641000, China (e-mail: zzc910731@163.com).}}

\markboth{Journal}  %
{How to Use the IEEEtran \LaTeX \ Templates}  
\maketitle

\begin{abstract}
Vertex-frequency analysis, particularly the windowed graph Fourier transform (WGFT), is a significant challenge in graph signal processing. 
Tight frame theories is known for its low computational complexity in signal reconstruction, 
while fractional order methods shine at unveil more detailed structural characteristics of graph signals.       
In the graph fractional Fourier domain, we introduce multi-windowed graph fractional Fourier frames (MWGFRFF) to facilitate the construction of tight frames. 
This leads to developing the multi-windowed graph fractional Fourier transform (MWGFRFT), enabling novel vertex-frequency analysis methods. 
A reconstruction formula is derived, along with results concerning dual and tight frames.
To enhance computational efficiency, a fast MWGFRFT (FMWGFRFT) algorithm is proposed.
Furthermore, we define shift multi-windowed graph fractional Fourier frames (SMWGFRFF) and their associated transform (SMWGFRFT), exploring their dual and tight frames.
Experimental results indicate that FMWGFRFT and SMWGFRFT excel in extracting vertex-frequency features in the graph fractional Fourier domain, with their combined use optimizing analytical performance. 
Applications in signal anomaly detection demonstrate the advantages of FMWGFRFT.
\end{abstract}


\section{Introduction}

\IEEEPARstart{S}{ignal} processing (SP) focuses on analyzing and processing information defined in Euclidean domains. 
However, with the rapid advancement of communication technologies and data proliferation \cite{Sandryhaila2014SPM}, modern applications often encounter information residing in non-Euclidean domains.             
SP techniques struggle to address the irregular nature of these domains.            
Graph signal processing (GSP) \cite{a1,a2,Sandryhaila2014TSP,a4,a5,a6} generalizes SP to signals defined on non-Euclidean domains, effectively capturing topological structures through weighted graphs.        
A central challenge in GSP lies in designing dictionaries of atoms and transform methods for accurately identifying and utilizing graph topology.    
Existing approaches primarily rely on graph adjacency matrices rooted in algebraic signal processing and Laplacian matrices derived from spectral graph theory \cite{a7}.         
Currently, the GSP theory mainly includes the following: graph Fourier transform (GFT) \cite{Sandryhaila2014SPM,b2}, graph filters \cite{b3,b4}, graph sampling and recovery \cite{b5,b6}, frequency analysis \cite{Sandryhaila2014TSP}, and fast algorithms \cite{Jestrovic2017sp,b8,b9}.         
These results have been applied to social networks and sensor networks and extended to machine learning.
     
Vertex-frequency analysis is a significant challenge in GSP.
Due to the classical Fourier transform's limitation in capturing time-varying properties, the windowed Fourier transform, also known as short-time Fourier transform, has become a crucial time-frequency analysis tool in SP.
Similarly, because GFT is ineffective for vertex-frequency representation, the windowed graph Fourier transform (WGFT) \cite{Shuman2016acha,c2,c3,c4,c5,c6,c7} is proposed.
Various operators related to WGFT are discussed \cite{Shuman2016acha,c3}, including the convolution operator, modulation operator and translation operator.
\cite{Shuman2016acha} also construct windowed graph Fourier frames (WGFF), i.e., dictionaries of atoms adapted to the underlying graph structure, facilitating effective vertex-frequency analysis.  
If WGFF is tight, the spectrogram can be interpreted as an energy density function of the signal over the vertex-frequency plane \cite{Shuman2016acha}.

Leveraging the knowledge of frames \cite{Casazza2012Springer}, we understand that signal reconstruction requires an additional dual frame calculation.        
By applying tight frames, the need for dual frame calculations can be eliminated, thereby saving the extra computational cost associated with dual frames.
To facilitate the construction of tight frames, multi-windowed graph Fourier frames (MWGFF) \cite{Zheng2021cc} are introduced to develop novel vertex-frequency analysis methods.  
The MWGFF and shift multi-windowed graph Fourier frames are utilized to extract vertex-frequency features of signals in the graph Fourier domain.                           
However, MWGFF is typically ineffective in processing signals in the graph fractional Fourier domain.       
        
To obtain more detailed structural properties of graph signals, a fractional order \cite{d1,d2,d3,d4,d5,d6,Wu2020mpe} is adopted.                               
Fractional graph signal processing is a nascent field that extends GSP techniques by incorporating fractional orders.
The graph fractional Fourier domain— a combination of the fractional transform domain and graph spectral domain— has typically been explored as a potentially enriching research topic.
GFRFT \cite{AKoc2024TSP} and SGFRFT \cite{Wu2020mpe} are advantageous for revealing local features of graph signals.
The windowed graph fractional Fourier transform (WGFRFT) \cite{Yan2021dsp} is introduced to extract vertex-frequency information from signals on weighted graphs. However, WGFRFT cannot effectively extract vertex-frequency information for a fixed window function.
                                
\begin{figure}[!t]
	\centering
	\includegraphics[width=0.6\columnwidth]{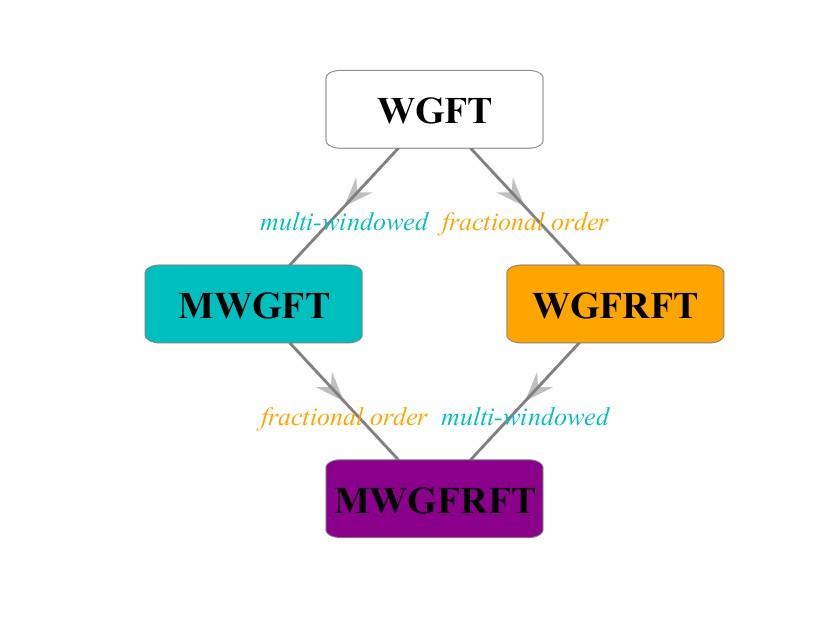}
	\caption{Relationships among WGFT \cite{Shuman2016acha}, MWGFT \cite{Zheng2021cc}, WGFRFT \cite{Yan2021dsp}, and MWGFRFT.}
	\label{figint001}
\end{figure}

In summary, traditional vertex-frequency analysis methods fail to extract vertex-frequency features in the graph fractional Fourier domain for given window functions.
To overcome these limitations, we propose the multi-windowed graph fractional Fourier transform (MWGFRFT), which combines the advantages of multi-window and fractional orders.
MWGFRFT is a generalized transform that includes WGFT \cite{Shuman2016acha}, MWGFT \cite{Zheng2021cc}, and WGFRFT \cite{Yan2021dsp}, as illustrated in Fig.~\ref{figint001}.
With its additional parameters $l$ and $\alpha$, MWGFRFT provides enhanced flexibility and adaptability.
The main contributions of our work can be summarized as follows:
\begin{itemize}
	\item 
	Based on the graph fractional translation operator and the graph fractional modulation operator, we obtain a multi-windowed graph fractional Fourier frame (MWGFRFF).
	Similarly, by applying the graph fractional shift operator and the graph fractional modulation operator, we present a shift multi-windowed graph fractional Fourier frame (SMWGFRFF).
	\item 
	The MWGFRFT and the shift multi-windowed graph fractional Fourier transform (SMWGFRFT) are defined to obtain a vertex-frequency representation in the graph fractional Fourier domain for given window functions.
	\item 
	We propose a fast MWGFRFT (FMWGFRFT) algorithm designed to significantly enhance computational efficiency.
	\item 
	We derive dual frames and tight frames corresponding to MWGFRFF and SMWGFRFF, and also provide a reconstruction formula.
	\item 
	In the graph fractional Fourier domain, we demonstrate an application of signal anomaly detection by applying WGFT \cite{Shuman2016acha}, MWGFT \cite{Zheng2021cc}, WGFRFT \cite{Yan2021dsp}, SMWGFT, SMWGFRFT, and FMWGFRFT. 
	The results highlight the effectiveness of FMWGFRFT in accurately identifying anomalous vertices.
\end{itemize}

The main connections and differences between the current study and previous ones are concluded as follows:
\begin{itemize}
	\item 
	These studies focus on vertex-frequency representations, aiming to extract detailed vertex and frequency features.
	\item
	WGFT \cite{Shuman2016acha}, MWGFT \cite{Zheng2021cc}, and WGFRFT \cite{Yan2021dsp}
	are special cases of the current study. 
	The current study addresses the limitation of WGFT, i.e., its inability to extract detailed vertex or frequency features;
	The current study overcomes the limitation of MWGFT, which cannot effectively perform frequency feature extraction;
	The current study resolves the limitation of WGFRFT, i.e., its inability to effectively extract vertex features.
\end{itemize}  

The remaining sections are organized as follows. 
Section \ref{section2} provides a brief review of the background relevant to this study. 
In Section \ref{section3}, we define MWGFRFF and present several results related to tight frames and introduce the concept of MWGFRFT. 
In Section \ref{section4}, we propose a fast MWGFRFT (FMWGFRFT) algorithm. 
Section \ref{section5} presents results concerning the dual of MWGFRFF and a reconstruction formula.
In Section \ref{section6}, we define SMWGFRFF and discuss related dual and tight frames. In addition, we introduce the SMWGFRFT concept. 
In Section \ref{section7}, the construction of tight MWGFRFF and SMWGFRFF is exemplified. 
In Section \ref{section8}, the experimental results are presented. 
Section \ref{section9} presents the applications related to anomaly detection. 
Finally, Section \ref{section10} concludes this study. 	 

\section{ Preliminaries}\label{section2}

\subsection{  Spectral graph theory}

An undirected weighted graph $\mathcal{G}=\{\mathcal{V},\mathcal{E},W\}$ consists of a finite set of vertices $\mathcal{V}$,
where $N$ is the number of vertices, $\mathcal{E}$ is a set of edges and $W$ is a weighted adjacency matrix.
The non-normalized graph Laplacian is a symmetric difference operator $\mathcal{L}=D-W$ \cite{Shuman2013SPM}, where $D$ is a diagonal degree matrix.  
Suppose that the corresponding Laplacian eigenvalues are arranged in ascending order as: $0=\lambda_{0}<\lambda_{1}\leq\lambda_{2}\leq\cdots\leq\lambda_{N-1}:=\lambda_{max}.$ 
Therefore,
$\mathcal{L}=\chi\Lambda\chi^H,$
where the superscript $H$ denotes the Hermitian transpose operation, $\chi=[\chi_0,\chi_1,\cdots,\chi_{N-1}]$ is a unitary column matrix
and $\Lambda=\operatorname{diag}([\lambda_0,\lambda_1,\cdots,\lambda_{N-1}])$ is a diagonal matrix.

For a signal $\mathbf{f}=[f(1),f(2),\cdots,f(N)]^T$ defined on the graph $\mathcal G$,
the GFT of $\mathbf{f}$ is
$$\widehat{f}(\lambda_k)=\langle \mathbf{f},\chi_{k}\rangle=\sum_{n=1}^{N}f(n)\chi_{k}^{*}(n),k=0,1,\cdots,N-1.$$
The inverse GFT is given by
$$f(n)=\sum_{k=0}^{N-1}\widehat{f}(\lambda_k)\chi_k(n),n=1,2,\cdots,N.$$

We observe that GFT is represented by the matrix $\chi$; 
as an extension, the spectral graph fractional Fourier transform (SGFRFT) is represented as $\chi^\alpha$.
The graph fractional Laplacian operator  $\mathcal{L}_\alpha=\gamma R\gamma^H$, where $0<\alpha\leq1$ \cite{Wu2020mpe}. 
Note that $\chi$ is unitary then
$\gamma=[\gamma_{0},\gamma_{1},\cdots,\gamma_{N-1}]
=\chi^{\alpha}$
is also unitary. 
And    
$R=\operatorname{diag}([r_0, r_1, \cdots, r_{N-1}])=\Lambda^\alpha,$
so that $r_k=\lambda_k^\alpha$ for $k=0,1,\cdots,N-1$.
The SGFRFT of a signal $\mathbf{f}$ defined on the graph $\mathcal{G}$ is expressed as \cite{Wu2020mpe}:
$$\widehat{f}_{\alpha}(r_k)=\sum_{n=1}^{N}f(n)\gamma_{k}^{*}(n),k=0,1,\cdots,N-1,$$
when $\alpha=1$, the SGFRFT degenerates into GFT.
The inverse SGFRFT is given by
$$f(n)=\sum_{k=0}^{N-1}\widehat{f}_\alpha(r_k)\gamma_k(n),n=1,2,\cdots,N.$$
The Parseval relation for the SGFRFT holds. For any signals $\mathbf{f}$ and $\mathbf{g}$ defined on the graph $\mathcal{G}$, we have
$\langle \mathbf{f},\mathbf{g}\rangle=\langle\widehat{\mathbf{f}}_{\alpha}, \widehat{\mathbf{g}}_{\alpha}\rangle.$

\subsection{ Windowed graph fractional Fourier transform\cite{Yan2021dsp}}

\begin{definition}(Graph fractional translation operator) 
	For any signal $\mathbf{g}$ defined on the graph $\mathcal{G}$ and $i\in\{1,2,\cdots,N\}$, the graph fractional translation operator $T_i^{\alpha}$ is defined as
	\begin{equation}{\label{eq1}}
		(T_i^\alpha \mathbf{g})(n)=(\sqrt N)^\alpha 
		\sum_{p=0}^{N-1}\hat{\mathbf{g}}_{\alpha}(r_p) \gamma_{p}^{*}(i)\gamma_{p}(n).
	\end{equation}
\end{definition}

\begin{definition}(Graph fractional modulation operator)
	For any signal $\mathbf{g}$ defined on the graph $\mathcal{G}$ and  $k\in\{0,1,\cdots,N-1\}$, define the graph fractional modulation operator by
	$$(M_k^\alpha \mathbf{g})(n)=(\sqrt{N})^\alpha \mathbf{g}(n)\gamma_k(n).$$
\end{definition}

The windowed graph fractional Fourier atoms generated by a window function $\mathbf{g}$ on $\mathcal{G}$ are defined by
$$\mathbf{g}_{i,k}^{(\alpha)}(n)
=(M_k^\alpha T_i^\alpha \mathbf{g})(n)
=N^\alpha\gamma_k(n)\sum_{p=0}^{N-1}\hat{\mathbf{g}}_\alpha(r_p)\gamma_p^*(i)\gamma_p(n),$$
where $i=1,2,\ldots,N$ and $k=0,1,\ldots,N-1$.
We denote the set of windowed graph fractional Fourier atoms by
\begin{equation}\label{eqwgfa1}
	\mathcal{G}_{\alpha}^{w}=\{\mathbf{g}_{i,k}^{(\alpha)}\}_{i=1,2,\ldots,N;k=0,1,\ldots,N-1;0<\alpha\leq1}.\end{equation}

\begin{definition}(Windowed graph fractional Fourier transform)
	Given a window function $\mathbf{g}\in\mathbb{C}^N$, the SGFRFT of $\mathbf{g}$ is $\hat{\mathbf{g}}_{\alpha}$.
	For a signal $\mathbf{f}$, 
	the WGFRFT of $\mathbf{f}$ is denoted by
	$$Wf(\mathbf{g}_{i,k}^{(\alpha)})=N^{\alpha}\sum_{n=1}^{N}f(n)(\sum_{p=0} ^{N-1}\hat{\mathbf{g}}_{\alpha}^{*}(r_p) \gamma_{p}(i)\gamma_{p}^{*}(n))\gamma_{k}^{*}(n).$$
\end{definition} 

\subsection{  Frames and operators\cite{Casazza2012Springer}}

A frame is a collection of vectors that span the signal space, providing a redundant representation of signals. This redundancy can be utilized to facilitate signal reconstruction.

\begin{definition}
	A family of vectors 
	$\{ \mathbf{f}_k\} _{k= 1}^M\subseteq \mathbb{C}^N \left( M\geq N\right)$ 
	is a finite frame for $\mathbb{C}^{N}$
	if there exist constants $A,B>0$ such that
	$$A\|\mathbf{f}\|^2\leq \sum_{k=1}^M|\langle\mathbf{f},\mathbf{f}_k\rangle|^2
	\leq B\|\mathbf{f}\|^2$$
	holds for every $\mathbf{f}\in\mathbb{C}^N.$
\end{definition}
The constants $A$ and $B$ are called frame bounds,
and if $A=B$, $\{\mathbf{f}_k\}_{k=1}^M$ is called a tight frame.
For a tight frame $\{\mathbf{f}_k\}_{k=1}^M$ with frame bound $C_0$, the following reconstruction formula holds:
$$\mathbf{f}=\frac{1}{C_0}\sum_{k=1}^M\langle\mathbf{f},\mathbf{f}_k\rangle \mathbf{f}_k, ~\mathbf{f}\in\mathbb{C}^N.$$

For a frame $\{\mathbf{f}_k\}_{k=1}^M$, define its synthesis operator by
$T(\{c_k\})=\sum\limits_{k=1}^Mc_k\mathbf{f}_k$
and the analysis operator by
$T^*\mathbf{f}=\{\langle\mathbf{f},\mathbf{f}_k\rangle\}.$
Integrating $T$ and $T^*$ together, which forms a pair of dual operators,
define the frame operator of $\{\mathbf{f}_k\}$ by $S=TT^*$, i.e.
$$S\mathbf{f}=\sum_{k=1}^M\langle\mathbf{f},\mathbf{f}_k\rangle\mathbf{f}_k,
~\mathbf{f}\in\mathbb{C}^N,$$
which is a positive, bounded, invertible and self-adjoint operator.

If there exists another sequence of vectors $\left\{\mathbf{g}_k\right\}_{k=1}^M$ such that
$$\mathbf{f}=\sum\limits_{k=1}^M\langle\mathbf{f},\mathbf{g}_k\rangle\mathbf{f}_k, \ \mathbf{f}\in\mathbb{C}^N,$$
then $\left\{\mathbf{g}_{k}\right\}_{k=1}^{M}$ is called a dual frame of $\left\{\mathbf{f}_{k}\right\}_{k=1}^{M}.$

\section{ Multi-windowed graph fractional Fourier frames}\label{section3}

In this section, we introduce multi-windowed graph fractional Fourier frames (MWGFRFF) and discuss related tight frames. Furthermore, we define multi-windowed graph fractional Fourier transforms (MWGFRFT).

Analog to \eqref{eqwgfa1}, for a finite sequence of window functions
$\mathbf{g}_1,\mathbf{g}_2,\ldots ,\mathbf{g}_L \in \mathbb{C}^N$, the SGFRFT of $\mathbf{g}_l$ is $\hat{\mathbf{g}}_{l\alpha}$,
we define the set of multi-windowed graph fractional Fourier atoms by
\begin{equation}\label{eqma1}
	\mathcal{G}_{l\alpha}^{w}=\{\mathbf{g}_{i,k}^{(l\alpha)}\}_{i=1,2,\ldots,N;k=0,1,\ldots,N-1;l=1,2,\ldots,L;0<\alpha\leq1},
\end{equation}
where
$$\mathbf{g}_{i,k}^{(l\alpha)}(n)
=(M_k^\alpha T_i^\alpha \mathbf{g}_l)(n)
= N^{\alpha}\gamma_{k}(n) \sum_{p=0}^{N-1}\hat{\mathbf{g}}_{l\alpha}(r_p) \gamma_{p}^{*}(i)\gamma_{p}(n).$$

\begin{theorem}\label{theorem1}
	Let $\mathcal{G}_{l\alpha}^{w}$ be the set of multi-windowed graph fractional Fourier atoms defined in \eqref{eqma1}. 
	If $\sum_{l=1}^{L}\left|\hat{\mathbf{g}}_{l\alpha}(0)\right|^{2}\neq0$, 
	then $\mathcal{G}_{l\alpha}^{w}$ is a frame (i.e., MWGFRFF), 
	thus for any signal $\mathbf{f}\in\mathbb{C}^N$,
	$$0<A\|\mathbf{f}\|_2^2\leq\sum_{l=1}^{L}\sum_{i=1}^{N}\sum_{k=0}^{N-1}|\langle\mathbf{f},\mathbf{g}_{i,k}^{(l\alpha)}\rangle|^{2}\leq B\|\mathbf{f}\|_2^2<\infty,$$
	where lower frame bound
	\begin{equation}\label{lfb1}
		0<A=\min_{n\in\{1,2,\cdots,N\}}\{N^\alpha\sum_{l=1}^{L}\|T_{n}^\alpha\mathbf{g}_{l}\|_2^2\},\end{equation}
	and upper frame bound
	\begin{equation}\label{lfb2}
		B=\max_{n\in\{1,2,\cdots,N\}}\{N^\alpha\sum_{l=1}^{L}\|T_{n}^\alpha\mathbf{g}_{l}\|_2^2\}<\infty.\end{equation}
\end{theorem}

\begin{proof}
 See Appendix \ref{atheorem1}.
\end{proof}

Let $\mathbf{c}=(c_{1},\ldots,c_{N})^{T},$ 
where $c_{n}=N^\alpha\sum\limits_{l=1}^{L}\|T_{n}^\alpha\mathbf{g}_{l}\|_{2}^{2}$. 
By Appendix \ref{atheorem1}, Eq.~\eqref{eqthp1} shows that
$\langle S\mathbf{f},\mathbf{f}\rangle
=\langle\mathbf{c\circ f},\mathbf{f}\rangle,\mathbf{f}\in\mathbb{C}^{N},$
where $\circ$ denotes the entrywise product.
Equivalently, we have
\begin{equation}\label{eqsfc1}
	S\mathbf{f}=\mathbf{c\circ f}
	=\mathbf{D}_{c}\mathbf{f},~\mathbf{f}\in\mathbb{C}^{N},
\end{equation}
where $\mathbf{D}_{c}=diag(\mathbf{c})$ is a diagonal matrix, with its $n$th diagonal entry $D_{nn}=c_{n}$.

\begin{corollary}\label{cor32}
	$\mathcal{G}_{l\alpha}^{w}$ is a tight frame if and only if there exists a constant $C$ such that $N^\alpha\sum_{l=1}^{L}\|T_{n}^\alpha\mathbf{g}_{l}\|_{2}^{2}=C$ for $n= 1, 2,\ldots,N$. \end{corollary}
\begin{proof}
	See Appendix \ref{acor32}.
\end{proof}

\begin{corollary}\label{cor33}
	Let $\mathcal{G}_{l\alpha}^{w}$ be the set of multi-windowed graph fractional Fourier atoms defined in \eqref{eqma1}. If there exists a constant $C$ such that 
	$\sum_{l=1}^{L}\left|\hat{\mathbf{g}}_{l\alpha}(r_{p})\right|^{2}=C$ 
	for $p= 0, 1, \ldots , N- 1$, 
	then $\mathcal{G}_{l\alpha}^{w}$ is a tight frame with frame bounds $A=B=N^{2\alpha}C.$ \end{corollary}
\begin{proof}
	See Appendix \ref{acor33}.
\end{proof}

\begin{definition}\label{mwgfft1}
	(Multi-windowed graph fractional Fourier transform)
	Given window functions $\mathbf{g}_1,\ldots,\mathbf{g}_L\in\mathbb{R}^N$, the SGFRFT of $\mathbf{g}_l$ is $\hat{\mathbf{g}}_{l\alpha}$.
	Let $\mathcal{G}_{l\alpha}^{w}$ is a multi-windowed graph fractional Fourier frame defined in Theorem \ref{theorem1}. 
	For a signal $\mathbf{f}\in\mathbb{C}^N$, 
	the MWGFRFT 
	$Wf(\mathbf{g}_{i,k}^{(l\alpha)})
	:=\langle\mathbf{f},\mathbf{g}_{i,k}^{(l\alpha)}(n)\rangle$ of $\mathbf{f}$ is
	\begin{equation}\label{eqMW1}Wf(\mathbf{g}_{i,k}^{(l\alpha)})=
		N^{\alpha}\sum_{l=1}^{L}\sum_{n=1}^{N}f(n)(\sum_{p=0} ^{N-1}\hat{\mathbf{g}}_{l\alpha}^{*}(r_p) \gamma_{p}(i)\gamma_{p}^{*}(n))\gamma_{k}^{*}(n),
	\end{equation}
\end{definition}
\noindent when $l=1$, Eq.~\eqref{eqMW1} degenerates into WGFRFT \cite{Yan2021dsp};
when $\alpha=1$, Eq.~\eqref{eqMW1} degenerates into MWGFT \cite{Zheng2021cc};
when $l=\alpha=1$, Eq.~\eqref{eqMW1} degenerates into WGFT \cite{Shuman2016acha}.

\begin{proposition}\label{ledsvf1} 
	(Diagonal clustering of the vertex-frequency representation with fractional order $\alpha\to 0$)
	When $\alpha$ tends to 0, the vertex-frequency representation of MWGFRFT is clustered on the main diagonal, i.e., all values outside the main diagonal equal to 0.
\end{proposition}
\begin{proof}
	See Appendix \ref{aledsvf1}.
\end{proof}
Note that Proposition \ref{ledsvf1} is also true for FMWGFRFT, MWGFT \cite{Zheng2021cc}, WGFRFT \cite{Yan2021dsp} and WGFT \cite{Shuman2016acha} and so on.

\section{ Fast multi-windowed graph fractional Fourier transform}\label{section4}

To improve computational speed, we introduce a fast multi-windowed graph fractional Fourier transform (FMWGFRFT) algorithm. 
In the classic scenario of signals, the windowed Fourier transform can be calculated through an inverse Fourier transform from the `$\alpha$-domain'(the Fourier transform of the windowed Fourier domain)\cite{Brown2009TSP}.
Applying the idea to graphs \cite{Yan2021dsp}, we define the $G^\alpha$-domain as the  SGFRFT of MWGFRFT in graph fractional Fourier domain:
\begin{equation}\label{eqGa1}
	G^\alpha(k,{k'})=\sum_{l=1}^L \sum_{i=1}^N Wf(\mathbf{g}_{i,{k'}}^{(l\alpha)})\gamma_{k}^*(i).
\end{equation}
Applying \eqref{eqMW1} to \eqref{eqGa1}, similar to \cite{Yan2021dsp}, then we get
$$G^\alpha(k,{k'})=N^\alpha\tilde{f}({k},{k'})\hat{\mathbf{g}}_{l\alpha}^*(r_{k}),$$
where $\tilde{f}({k},{k'})=\sum\limits_{d=1}^Nf(d)\gamma_{k'}^*(d)\gamma_{k}^*(d)$.

\begin{definition}\label{demw6} 
	(Fast multi-windowed graph fractional Fourier transform). 
	The inverse SGFRFT of $G^\alpha$ on the variable $k$ is FMWGFRFT, i.e.,
	$$FWf(\mathbf{g}_{i,{k'}}^{(l\alpha)})
		=\sum_{l=1}^L\sum_{k=0}^{N-1}N^{\alpha}\tilde{f}({k},{k'})\hat{\mathbf{g}}_{l\alpha}^*(r_{k})\gamma_{k}(i). $$
\end{definition}

The calculation algorithm of FMWGFRFT is defined as follows:

(i) The loop about $\hat{\mathbf{g}}_{l\alpha}$ for $l=1,\ldots,L$.
This step has $O(L)$ computational complexity.

(ii) Calculate $\tilde{f}({k},{k'})$ by using matrices:
$$\tilde{F}=[\tilde{f}({k},{k'})]=(F\circ\gamma^H)*\gamma\:,$$
where $\circ$ represents Hadamard multiplication and $*$ denotes the standard matrix multiplication. 
Each row of $F$ is the signal $\mathbf{f}$, and $\gamma^H$ denotes the complex conjugate of the graph fractional basis $\gamma$.
This step has $O(N^3)$ computational complexity.

(iii) Form a matrix such that each row is the SGFRFT of a window function ${\mathbf{g}}_{l}$.
$$\Psi^{l\alpha}=\begin{pmatrix}
	\hat{\mathbf{g}}_{l\alpha}(r_0)&\hat{\mathbf{g}}_{l\alpha}(r_1)&\cdots&\hat{\mathbf{g}}_{l\alpha}(r_{N-1})\\
	\hat{\mathbf{g}}_{l\alpha}(r_0)&\hat{\mathbf{g}}_{l\alpha}(r_1)&\cdots&\hat{\mathbf{g}}_{l\alpha}(r_{N-1})\\
	\vdots&\vdots&\ddots&\vdots\\
	\hat{\mathbf{g}}_{l\alpha}(r_0)&\hat{\mathbf{g}}_{l\alpha}(r_1)&\cdots&\hat{\mathbf{g}}_{l\alpha}(r_{N-1})
\end{pmatrix}$$
This step does not influence the computational complexity.

(iv) Calculate the $G^\alpha$-domain representation.
$$G^\alpha=N^\alpha*(\Psi^{l\alpha})^*\circ\tilde{F}.$$
The computational complexity is $O\left(N^{2}\right)$.

(v) The vertex-frequency content is obtained by applying inverse SGFRFT to each row of $G^\alpha$,i.e.
$$FWf(l)=\gamma * (G^\alpha)^{.T},$$
where $(G^\alpha)^{.T}$ denotes the non-conjugate transpose of $G^\alpha$.
The complexity of this step is $O(N^3).$    

(vi) Perform matrix addition on these $L$ matrices.
The complexity of this step is $O(LN^2).$

To sum up, the computational complexity of FWGFRFT is $O(LN^3)$. 
As $N$ increases, this computational complexity is a significant improvement, 
compared with MWGFRFT algorithm and MWGFT \cite{Zheng2021cc} algorithm which have the computational complexity of $O(LN^4)$.

\section{ Dual of multi-windowed graph fractional Fourier frames}\label{section5}

In this section, we obtain a reconstruction formula by introducing the dual of multi-windowed graph fractional Fourier frames. 
In addition, the canonical dual is also introduced.

Let $\tilde{\mathcal{G}}_{l\alpha}^{w}$ denote the set of multi-windowed graph fractional Fourier atoms, generated by a finite sequence of window functions $\tilde{\mathbf{g}}_1,\tilde{\mathbf{g}}_2,\ldots,\tilde{\mathbf{g}}_N$ in the sense of \eqref{eqma1}. 
$\tilde{\mathcal{G}}_{l\alpha}^{w}$ is called a dual to ${\mathcal{G}}_{l\alpha}^{w}$ if for any $\mathbf{f}\in\mathbb{C}^N$, there exists a constant $C$ such that
\begin{align}\label{eq51}
	\mathbf{f}
	&=C\sum_{l=1}^{L}\sum_{k=0}^{N-1}\sum_{i=1}^{N}\langle\mathbf{f},\tilde{\mathbf{g}}_{i,k}^{(l\alpha)}\rangle\mathbf{g}_{i,k}^{(l\alpha)}
	\nonumber  \\
	&=C\sum_{l=1}^{L}\sum_{k=0}^{N-1}\sum_{i=1}^{N}\langle \mathbf{f},{\mathbf{g}}_{i,k}^{(l\alpha)}\rangle\tilde{\mathbf{g}}_{i,k}^{(l\alpha)}.\end{align}

\begin{theorem}\label{theorem2}
	Suppose that ${\mathcal{G}}_{l\alpha}^{w}$ is $a$ multi-windowed graph fractional Fourier frame as defined in \eqref{eqma1}. If there exists a finite sequence of window functions $\tilde{\mathbf{g}}_1,\tilde{\mathbf{g}}_2,\ldots,\tilde{\mathbf{g}}_L$ and a constant $\mu>0$ such that
	$$\sum_{l=1}^{L}\mathbf{\hat{g}}_{l\alpha}^{*}(r_{p})\mathbf{\hat{\tilde{g}}}_{l\alpha}(r_{p})=\mu,\:p=0,1,\ldots,N-1,$$
	then $\tilde{\mathcal{G}}_{l\alpha}^{w}$ is a dual of ${\mathcal{G}}_{l\alpha}^{w}$.
\end{theorem}
\begin{proof}
	See Appendix \ref{atheorem2}.
\end{proof}
We could also prove that the dual $\tilde{\mathcal{G}}_{l\alpha}^{w}$ of a multi-windowed graph fractional Fourier frame ${\mathcal{G}}_{l\alpha}^{w}$ is also a frame.

\begin{corollary}\label{cor42}
	Suppose that ${\mathcal{G}}_{l\alpha}^{w}$ is a multi-windowed graph fractional Fourier frame as defined in \eqref{eqma1}. 
	If there exists a finite sequence of window functions $\tilde{\mathbf{g}}_{1},\tilde{\mathbf{g}}_{2},...,\tilde{\mathbf{g}}_{L}$ and a constant $\mu>0$
	such that $\sum_{l=1}^{L}\mathbf{\hat{g}}_{l\alpha}^{*}(r_{p})\mathbf{\hat{\tilde{g}}}_{l\alpha}(r_{p})=\mu$ for $p=0,1,\ldots,N-1$, 
	then ${\mathcal{\tilde{G}}}_{l\alpha}^{w}$ is also a multi-windowed graph fractional Fourier frame.
\end{corollary} 
\begin{proof}
	See Appendix \ref{acor42}.
\end{proof}

\begin{corollary}\label{cor43}
	Suppose that ${\mathcal{G}}_{l\alpha}^{w}$ is a multi-windowed graph fractional Fourier frame as defined in \eqref{eqma1}. 
	Let $\mathbf{c}\in\mathbb{C}^N$ be a vector with $c_n=N^\alpha\sum_{l=1}^{L}\|T_{n}^\alpha\mathbf{g}_{l}\|_{2}^{2}$
	and
	$\mathbf{d}\in\mathbb{C}^{N}$ be a vector with $d_n=\frac1{c_n}$ for $n=1,2,\ldots,N$. 
	Then the canonical dual frame of ${\mathcal{\tilde{G}}}_{l\alpha}^{w}$ is given by
	$${\mathcal{\tilde{G}}}_{l\alpha}^{w}:=
	\{\mathbf{\tilde{g}}_{i,k}^{(l\alpha)}=
	\mathbf{d}\circ \mathbf{g}_{i,k}^{(l\alpha)}\},$$
	where $i= 1, 2, \ldots , N$; $k= 0, 1, \ldots , N- 1$; $l= 1, 2, \ldots , L$; $0<\alpha\leq1$.\end{corollary}
\begin{proof}
	See Appendix \ref{acor43}.
\end{proof}

\section{ Shift multi-windowed graph fractional Fourier frames} \label{section6}

In GSP, the adjacency matrix acts as a shift operator \cite{Sandryhaila2014SPM}. In graph fractional Fourier domain,
different fractional shift operators are proposed \cite{Yan2022ICSIPC,Yan2023ICSIPC,Ribeiro2022access}. 
Assume that the eigen-decomposition of the graph adjacency matrix is $\mathbf{A}=VJV^{-1}$,
in this paper, define the graph fractional shift operator as
$$S^\alpha=V J^\alpha V^{-1},$$
to facilitate the construction of tight frames in Section \ref{section7} later.
Then we could design a new type of frames and introduce the concept of shift multi-windowed graph fractional Fourier transform (SMWGFRFT).

\begin{theorem}\label{theorem3}
	Given $L$ window vectors $\mathbf{g}_{l}\in\mathbb{C}^{N}$ with $S^\alpha\mathbf{g}_{l}\neq0$ for $l=1,\ldots,L$.
	Define 
	$\mathbf{g}_{i,l}:=\gamma_i\circ(S^\alpha\mathbf{g}_{l}),$
	where $\circ$ denotes the Hadamard product. 
	Also define
	\begin{equation}\label{eqth31}
		\mathbf{c}=(c_1,c_2,\ldots,c_N)^T,\mathrm{with}\:c_k=\tilde{\mathbf{s}}_k\mathbf{G}\tilde{\mathbf{s}}_k^*,
	\end{equation}
	where $\tilde{\mathbf{s}}_k$ is the $k$th row of matrix $S^\alpha$ and 
	$\mathbf{G}= \sum_{i=1}^{L}\mathbf{g}_{i}\mathbf{g}_{i}^{*}$.
	The set of shift multi-window graph fractional Fourier atoms
	\begin{equation}\label{eqth34}
		\mathcal{G}_{l\alpha}^{s}:=\{\mathbf{g}_{i,l}\}_{i=0,1,\ldots,N-1;l=1,2,\ldots,L}\end{equation}
	forms a frame (i.e., SMWGFRFF) for signals defined on $\mathcal{G}$ if and only if $c_k> 0$ for all of elements of $\mathbf{c}$. 
	The optimal lower and upper frame bounds are
	\begin{equation}\label{eqth32}
		A=\min_{k\in\{1,2,...,N\}}c_k ~\text{and}\:B=\max_{k\in\{1,2,...,N\}}c_k,
	\end{equation}
	respectively.
\end{theorem}
\begin{proof}
	See Appendix \ref{atheorem3}.
\end{proof}

\begin{corollary}\label{cor52}
	Suppose that $\mathcal{G}_{l\alpha}^{s}$ is a shift multi-windowed graph fractional Fourier frame as defined in \eqref{eqth34}. Let $\mathbf{c}\in\mathbb{C}^N$ be a vector with $c_k=\tilde{\mathbf{s}}_k\mathbf{G}\tilde{\mathbf{s}}_k^*$ and $\mathbf{d}\in\mathbb{C}^{N}$ be a vector with $d_n= \frac 1{c_n}$ for $n= 1, 2, \ldots , N.$ Then the canonical dual frame of $\mathcal{G}_{l\alpha}^{s}$ is defined as
	$\tilde{\mathcal{G}}_{l\alpha}^{s}=\{\tilde{\mathbf{g}}_{i,l}=\mathbf{d}\circ\mathbf{g}_{i,l}\},$
	where $n= 1, 2, \ldots , N; i= 0, 1, \ldots , N- 1; l= 1, 2, \ldots , L$.
\end{corollary}
\begin{proof}
	The proof is similar to Corollary \ref{cor43}.
\end{proof}

\begin{corollary}\label{cor53}
	Let $\mathcal{G}_{l\alpha}^{s}$ be a shift multi-windowed graph fractional Fourier frame
	and
	$\tilde{\mathbf{s}}_k$ is the $k$th row of matrix $S^\alpha$. $\mathcal{G}_{l\alpha}^{s}$ is a tight frame if and only if
	$\tilde{\mathbf{s}}_k\mathbf{G}\tilde{\mathbf{s}}_k^*=C$ for $k=1,\ldots,N,$
	where C is a constant number and $\mathbf{G}=\sum_{i=1}^L\mathbf{g}_i\mathbf{g}_{i}^*$.
\end{corollary}
\begin{proof}
	See Appendix \ref{acor53}.
\end{proof}

It is known that 
$S^\alpha
=\begin{pmatrix}
	\tilde{\mathbf{s}}_1\\
	\vdots\\
	\tilde{\mathbf{s}}_N
\end{pmatrix}$, 
where
$\tilde{\mathbf{s}}_i$ is the $i$th row of matrix $S^\alpha$. 
Then we define
$S_i^\alpha=(\tilde{\mathbf{s}}_i)^{.T}$, 
where $(\tilde{\mathbf{s}}_i)^{.T}$ denotes the non-conjugate transpose of $\tilde{\mathbf{s}}_i$.
Given window functions $\mathbf{g}_1,\ldots,\mathbf{g}_L\in\mathbb{C}^N$,
we can define a new set of shift multi-windowed graph fractional Fourier atoms by
\begin{equation}\label{eqsma6}
	\mathcal{G}_{l\alpha}^{sw}=\{\mathbf{g}_{i,k}^{(l\alpha)}\}_{i=1,2,\ldots,N;k=0,1,\ldots,N-1;l=1,2,\ldots,L;0<\alpha\leq1},
\end{equation}
where
$$\mathbf{g}_{i,k}^{(l\alpha)}(n)
=(M_k^\alpha S_i^\alpha \mathbf{g}_l)(n).$$

\begin{definition}\label{s6smwgfft2}
	(Shift multi-windowed graph fractional Fourier transform)
	Given window functions $\mathbf{g}_1,\ldots,\mathbf{g}_L\in\mathbb{C}^N$,
	let $\mathcal{G}_{l\alpha}^{sw}$ is a shift multi-windowed graph fractional Fourier frame defined in \eqref{eqsma6}. 
	For a graph signal $\mathbf{f}\in\mathbb{C}^N$, 
	the SMWGFRFT of $\mathbf{f}$ is denoted by
	$$SWf(\mathbf{g}_{i,k}^{(l\alpha)})
	:=\langle\mathbf{f},\mathbf{g}_{i,k}^{(l\alpha)}(n)\rangle.$$
\end{definition}
It is apparently that the computational complexity of SMWGFRFT is $O(LN^3)$.

\section{ Two types of tight multi-windowed graph fractional Fourier frames}\label{section7}

In this section, we design tight multi-windowed graph fractional Fourier frames(TMWGFRFF) and tight shift multi-windowed graph fractional Fourier frames (TSMWGFRFF).

\subsection{  Construction of TMWGFRFF} \label{section7A}

Constructing tight MWGFRFF is equivalent to find $L$ window functions $\{\mathbf{g}_1,\mathbf{g}_2,\ldots,\mathbf{g}_{L}\}$ such that 
$\sum_{l=1}^{L}\|T_{n}^\alpha\mathbf{g}_{l}\|_2^2=C$
for $n=1,2,\ldots,N.$
By Corollary \ref{cor33}, 
the goal of the construction is to
find a sequence of window functions such that 
$\sum_{l=1}^{L}\left|\hat{\mathbf{g}}_{l\alpha}(r_{p})\right|^{2}=1$ for $r_p \in R$, where $R$ denotes graph fractional Laplacian spectrum.
Among the functions with such a property, we think of the cardinal B-spline as a candidate.

The $k$th order cardinal B-spline $N_k$ is defined as follows:
$$\mathbf{N}_k(x)=\int_0^1\mathbf{N}_{k-1}(x-t)dt.$$
The integer-translates of the $k$th order cardinal B-spline form a partition of unity, i.e.,
\begin{equation}\label{eq711}
	\sum_{p\in\mathbb{Z}}\mathbf{N}_{k}(x-p)=1,\forall \ k\in\mathbb{N}^{+}.\end{equation}
Suppose that the normalized graph fractional Laplacian matrix is applied, that is, the corresponding graph fractional Laplacian spectrum is contained in [0, 2]. 
We then take B-spline $\mathbf{N}_{2}$ to construct the tight frame window functions:
$$\left.
\mathbf{N}_{2}(x)=\left\{\begin{array}{ll}x,&\quad x\in[0,1),\\
	2-x,&\quad x\in[1,2],\\
	0,&\quad\text{else.}\end{array}\right.
\right.$$
Letting $| \hat{g_1}_{\alpha}|^{2}=\mathbf{N} _{2}( x- 1)$, 
$| \hat{g_2} _{\alpha}|^{2}=\mathbf{N} _{2}( x)$, 
and $| \hat{g_3} _{\alpha}| ^{2}$ = $\mathbf{N} _{2}( x+ 1)$, according to Eq.~\eqref{eq711}, we have $\sum\limits_{l= 1}^{3}| \hat{g_{l}} _{\alpha}( \lambda _{p}) | ^{2}= 1, \lambda _{p}\in [ 0, 2] .$That is, $\{g_1,g_2,g_3\}$ is a set of tight frame window functions.

\subsection{  Construction of TSMWGFRFF}  \label{section7B}

When the rows of the graph fractional shift operator $S^\alpha$ have identical 2-norm, we could construct orthonormal vectors as the tight SMWGFRFF:

(i) the graph fractional Laplacian eigenvectors, i.e. $\mathbf{g}_l= \gamma_{l-1}$ for $l= 1, 2, \ldots , N$, where $\gamma_l$ is the $l$th eigenvector of the graph fractional Laplacian matrix.

(ii) Householder vectors with localized generator, i.e. $\mathbf{g}_l=\mathbf{h}_l$, where $\mathbf{h}_{l}$ is the $l$th vector of the Householder matrix $\mathbf{H}=\mathbf{I}_{N}-2\mathbf{v}\mathbf{v}^{T}$ and $\mathbf{v}$ is a vector localized on a small set of indices.


\section{ Experiments}\label{section8}

We conduct several simulations by using frame, dual frame, and tight frame window functions to evaluate the effectiveness of MWGFRFT, FMWGFRFT, and SMWGFRFT in extracting the vertex-frequency features of graph signals. 
Our methods are also compared with WGFT\cite{Shuman2016acha}, MWGFT\cite{Zheng2021cc} and WGFRFT\cite{Yan2021dsp}.

\subsection{ Comparative Analysis of Vertex-Frequency Feature Extraction among WGFT\cite{Shuman2016acha}, MWGFT\cite{Zheng2021cc}, WGFRFT\cite{Yan2021dsp} and MWGFRFT}

First, we consider a path graph with 256 vertices, all of which have equal weights of 1. 
The eigenvectors of the graph Laplacian are generated from the basis vectors of the DCT-II transform.
We consider the following three signals in Fig.~\ref{figf123}: $f_1$ is a band limited signal in graph fractional Fourier domain, $f_2$ and $f_3$ are bandpass signals in graph fractional Fourier domain.
We choose the heat diffusion kernel as the window function by setting $\hat{g}_{\alpha}(\lambda_\ell) = Ce^{-\tau\lambda_\ell}$ with $\tau = 300$ and choosing $C$ such that $\|\hat{g}_{\alpha}\|_2 = 1$. 
Applying Eq.\eqref{eq1} to $\hat{g}_{\alpha}$, then obtain $L=10$ windows.
The ``spectrograms" in Fig.~\ref{figvf1} show $|Wf|^2$ for all $i \in \{1,2,\ldots,256\}$ and $k \in \{0,1,\ldots,255\}$.

From Fig.~\ref{figvf1}, the vertex-frequency features of the three signals are clearly observed using MWGFRFT, enabling precise identification of the brightest points by combining their vertex and frequency information.
However, none of WGFT, MWGFT and WGFRFT can accurately extract the vertex-frequency features of the three signals.
The vertex domain features of the three signals by WGFT are unclear and cannot be accurately determined.
For MWGFT, the vertex domain features of the 3 signals are clearly determined, however, the corresponding frequency features at each vertex cannot be accurately localized, nor can the location of the brightest point in the spectrograms be determined.
For WGFRFT, no accurate vertex features are identified, and all frequency features are located at the smallest vertices $i$.

\begin{figure*}[!t]
	\centering
	\subfloat[$f_1$
	]{\includegraphics[width=0.32\textwidth]{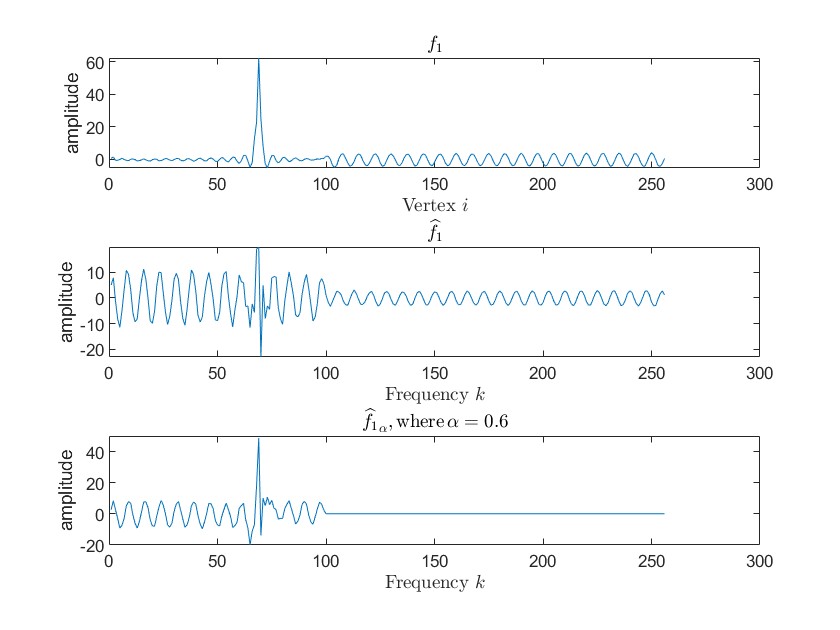}} \hfill
	\subfloat[$f_2$
	]{\includegraphics[width=0.32\textwidth]{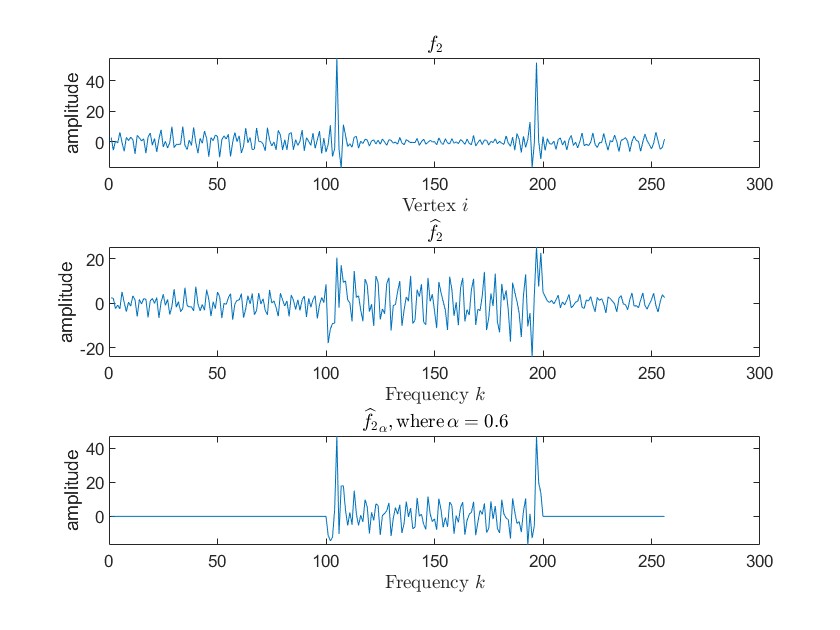}} \hfill
	\subfloat[$f_3$
	]{\includegraphics[width=0.32\textwidth]{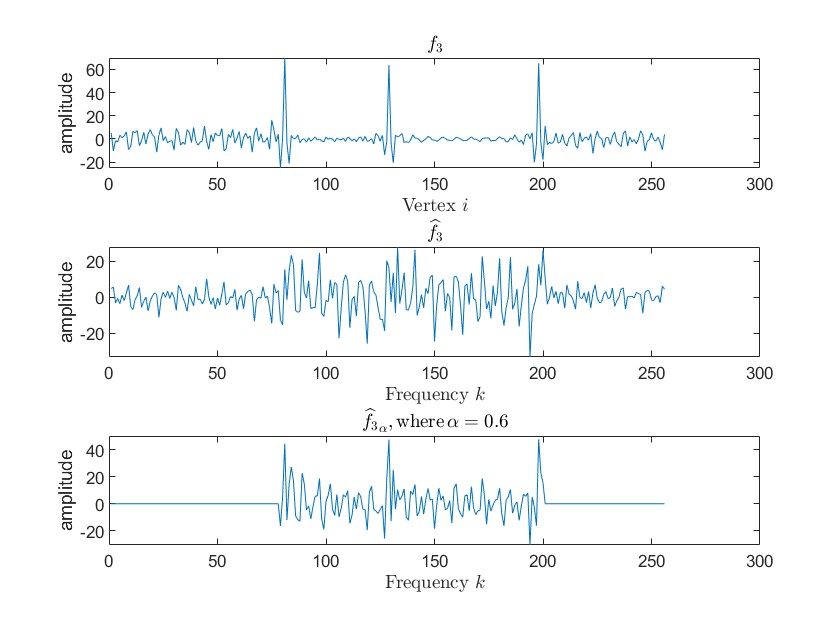}}
	\caption{The signal $f_1$, $f_2$, and $f_3$.}
	\label{figf123}
\end{figure*}

\begin{figure*}[!t]
	\centering
	\subfloat[WGFT of $f_1$ ]{\includegraphics[width=0.23\linewidth]{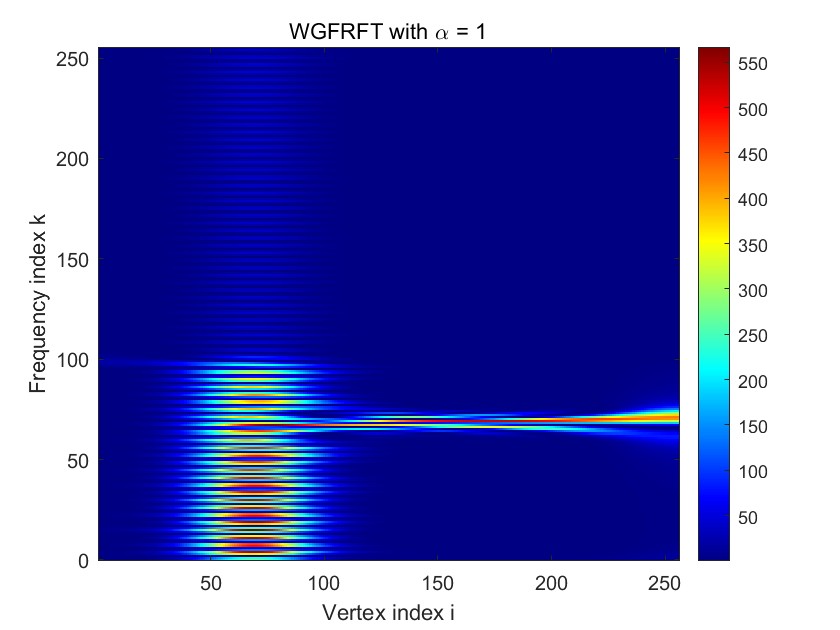}} \hfill
	\subfloat[MWGFT of $f_1$ ]{\includegraphics[width=0.23\linewidth]{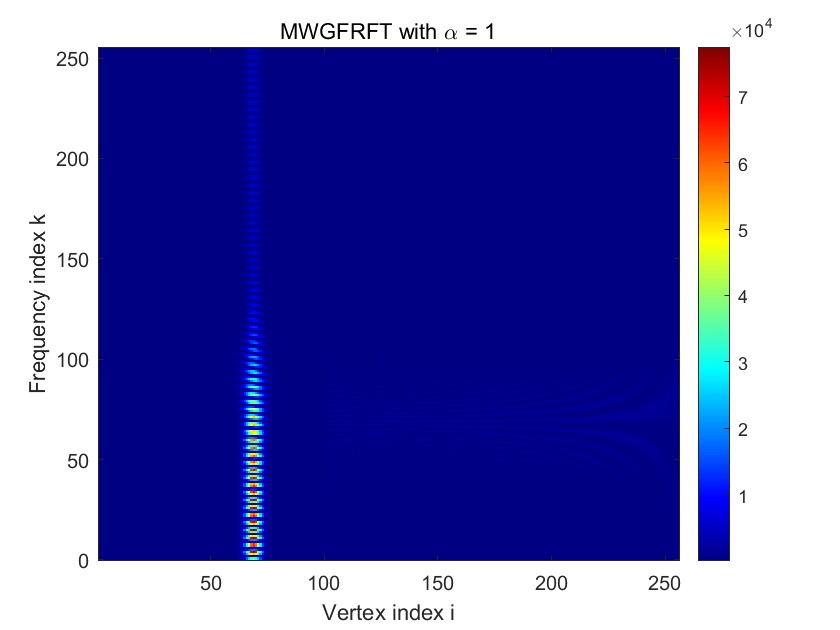}} \hfill
	\subfloat[WGFRFT of $f_1$
	]{\includegraphics[width=0.23\linewidth]{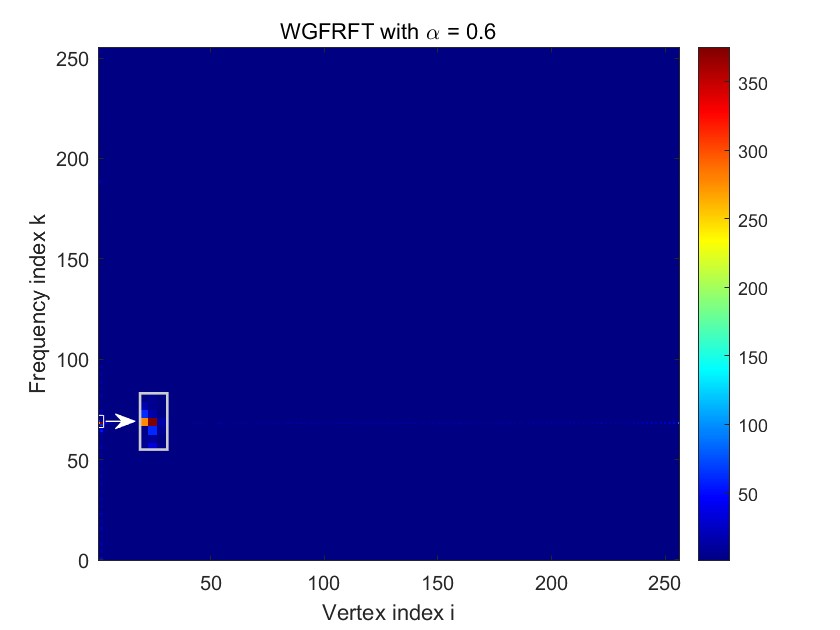}} \hfill
	\subfloat[MWGFRFT of $f_1$
	]{\includegraphics[width=0.23\linewidth]{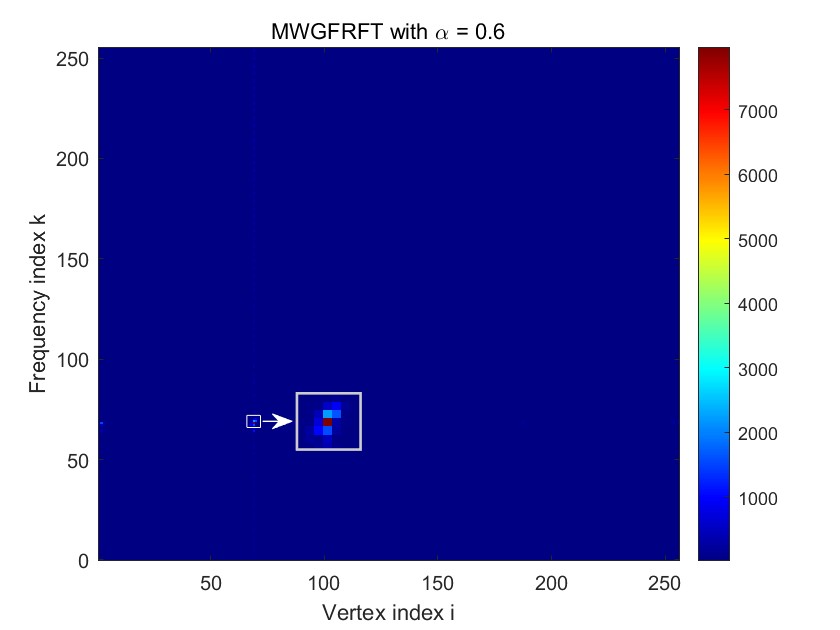}} \\
	\vspace{0.2cm}
	\subfloat[WGFT of $f_2$
	]{\includegraphics[width=0.23\linewidth]{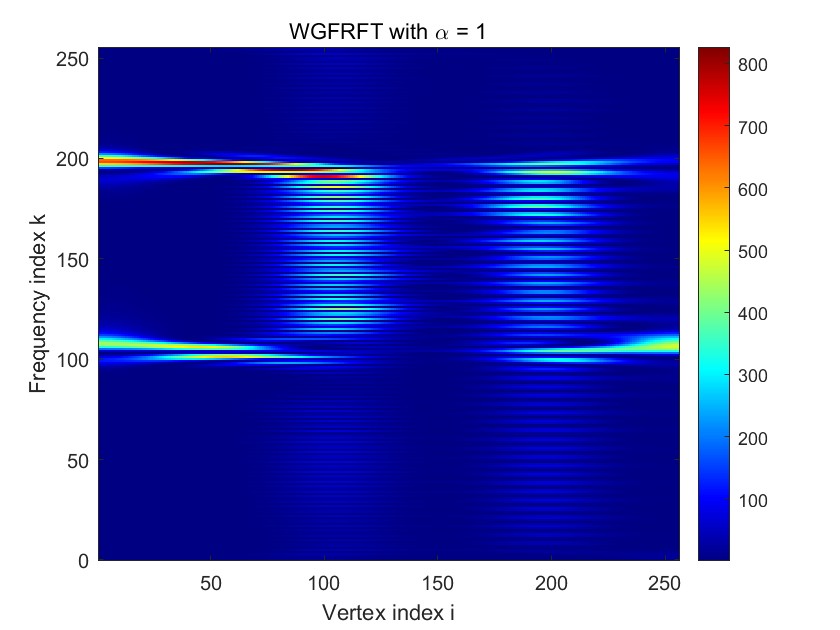}} \hfill
	\subfloat[MWGFT of $f_2$
	]{\includegraphics[width=0.23\linewidth]{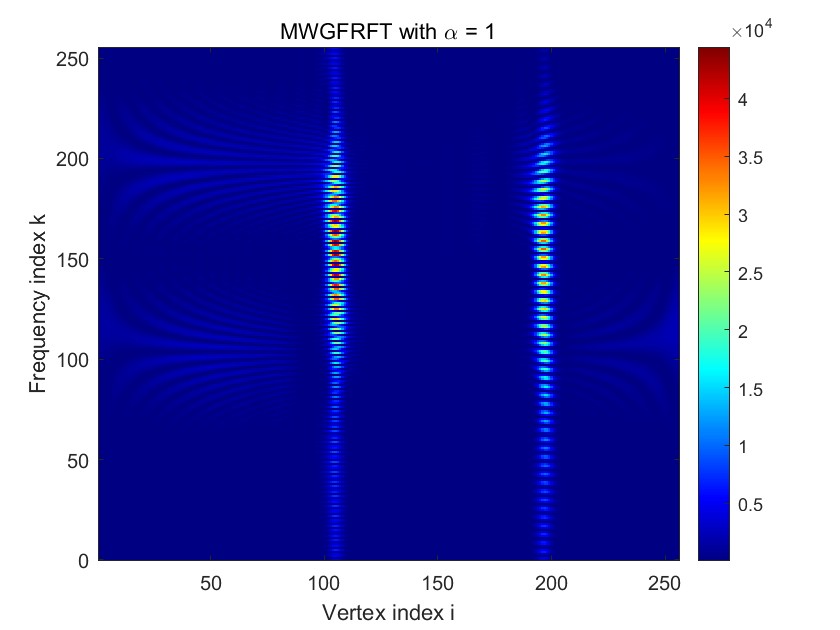}} \hfill
	\subfloat[WGFRFT of $f_2$
	]{\includegraphics[width=0.23\linewidth]{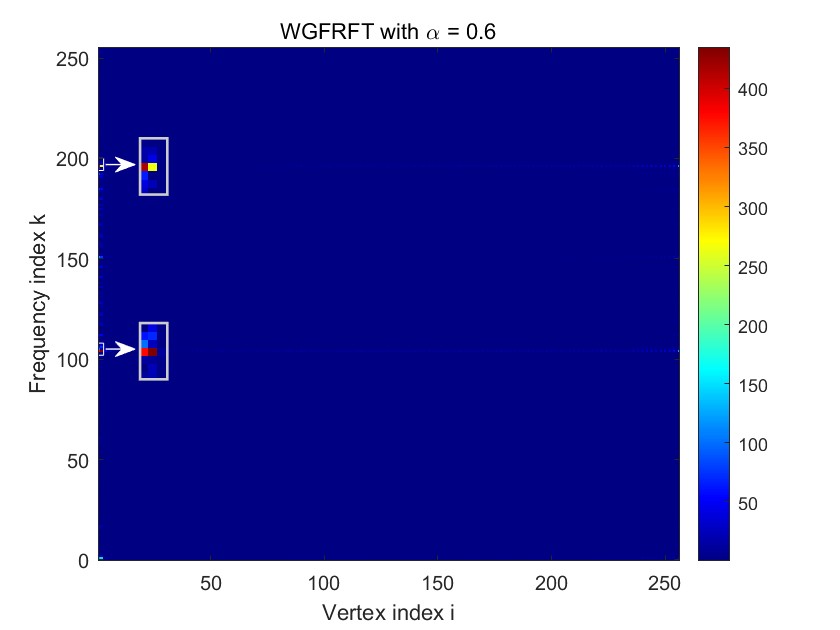}} \hfill
	\subfloat[MWGFRFT of $f_2$
	]{\includegraphics[width=0.23\linewidth]{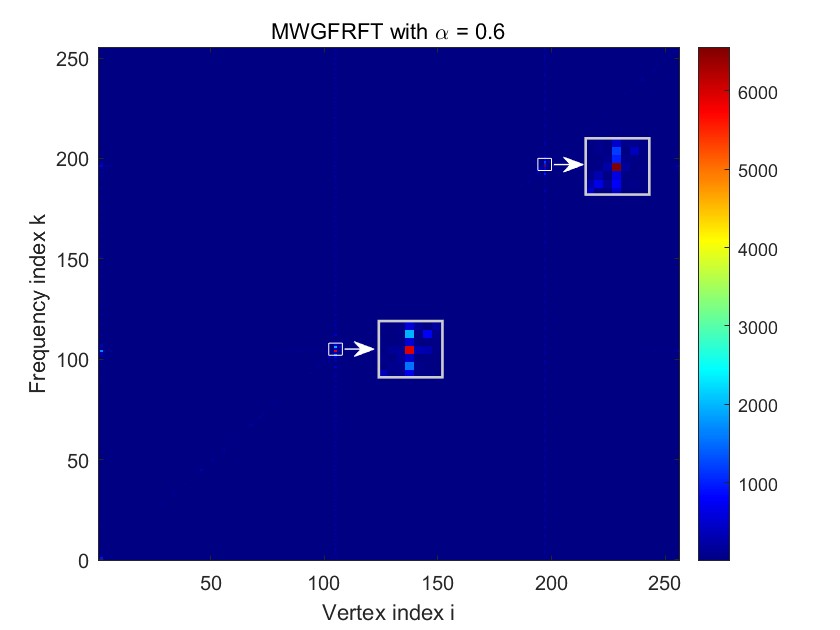}} \\
	\vspace{0.2cm}
	\subfloat[WGFT of $f_3$
	]{\includegraphics[width=0.23\linewidth]{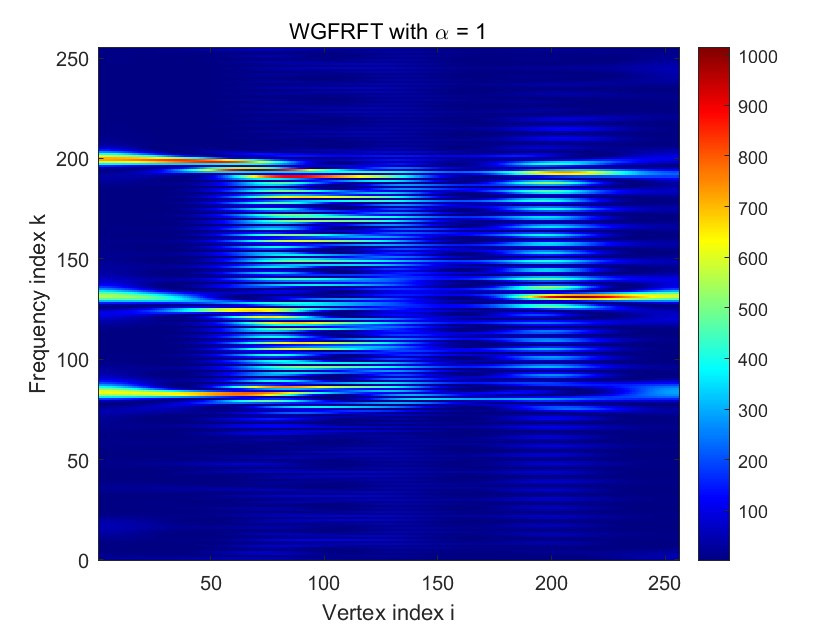}} \hfill
	\subfloat[MWGFT of $f_3$
	]{\includegraphics[width=0.23\linewidth]{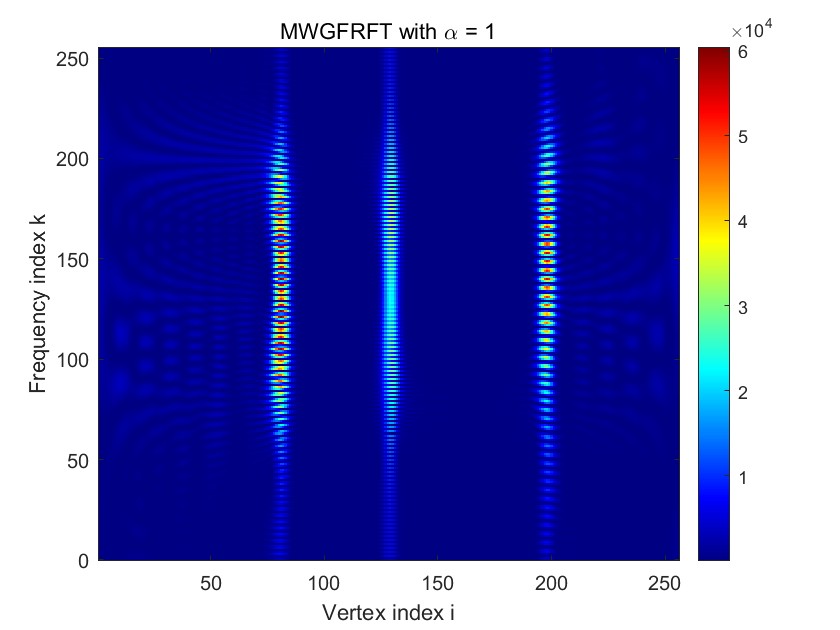}} \hfill
	\subfloat[WGFRFT of $f_3$
	]{\includegraphics[width=0.23\linewidth]{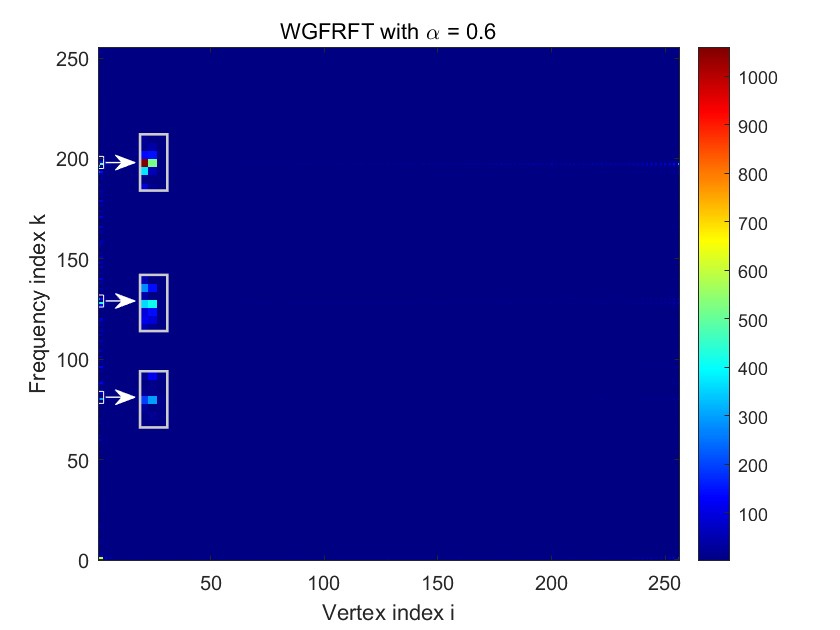}} \hfill
	\subfloat[MWGFRFT of $f_3$
	]{\includegraphics[width=0.23\linewidth]{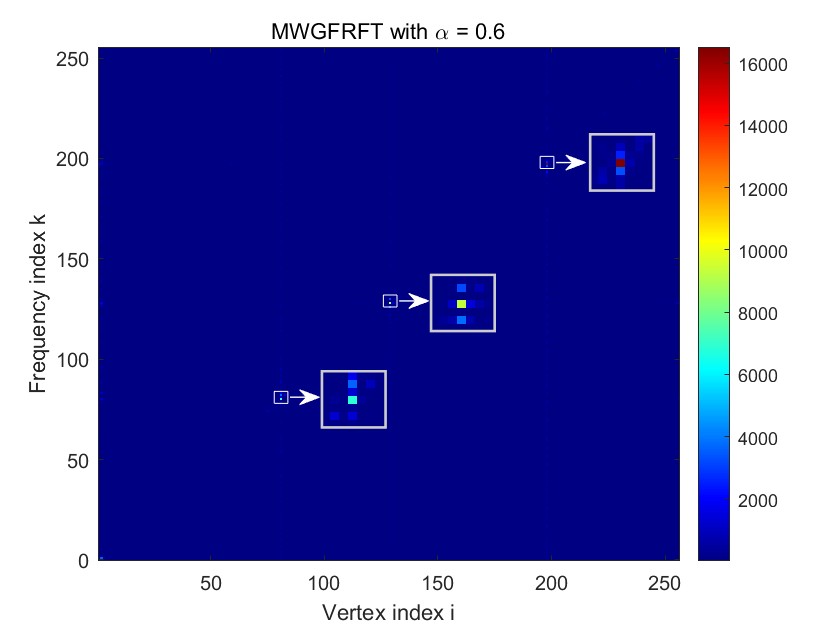}}
	\caption{Vertex-frequency representations of signals $f_1$, $f_2$, and $f_3$ by WGFT, MWGFT, WGFRFT and MWGFRFT, with fractional order $\alpha = 0.6$ 
		(Observe the six spectrograms (c-d), (g-h), and (k-l), where each larger rectangular box, marked by an arrow, presents a magnified view of the corresponding smaller rectangular box within the spectrogram).}
	\label{figvf1}
\end{figure*}

\subsection{  Algorithm evaluation of FMWGFRFT}

To evaluate the performance of FMWGFRFT algorithm, we conducted experiments to demonstrate its effectiveness and robustness.

(i) effectiveness

To clearly demonstrate the temporal advantage of the fast algorithm, we compare the time cost of the fast algorithm (FMWGFRFT) with the original algorithm (MWGFRFT).
Setting the fractional order $\alpha = 0.8$, we contemplate a random ring graph with different number of vertices $N$, where the signal is $f_5(n)$.
We choose the standard Gaussian function as the window function $\hat{g}_{\alpha}$,
applying Eq.\eqref{eq1} to $\hat{g}_{\alpha}$, and then obtain $L=10$ windows.
In Fig.~\ref{figtc3}, the x-axis represents the logarithm of the number of vertices and the y-axis represents the logarithm of computation time (in seconds).
For MWGFRFT, the red line corresponds to a linear fit, aligning with the theoretical computational complexity of \( O(LN^4) \). 
Similarly, for FMWGFRFT, the blue line corresponds to a linear fit, which matches the theoretical computational complexity of \( O(LN^3) \). 
The results indicate that FMWGFRFT is more efficient and requires less computation time compared to MWGFRFT.

\begin{figure}[!t]
	\centering
	\includegraphics[width=\columnwidth]{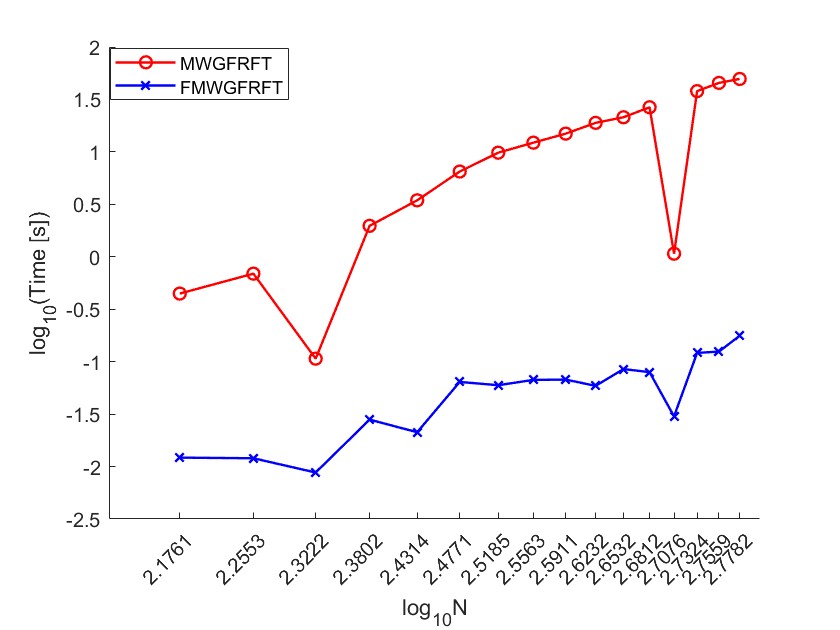}
	\caption{Time cost comparison of MWGFRFT and FMWGFRFT}
	\label{figtc3}
\end{figure}

(ii) Robustness

In the first experiment, we simulate with 3 graphs of 200 vertices: a ring graph, a random ring graph and a sphere graph. We take the real parts of the 3 graph signals as follows.
\begin{align*}
	f_4(n) &= Re(\gamma_{100}^{\text{r}}(n)), \\
	f_5(n) &= Re(\gamma_{50}^{\text{rr}}(n) + \gamma_{90}^{\text{rr}}(n) + \gamma_{150}^{\text{rr}}(n)), \\
	f_6(n) &= Re(\gamma_{64}^{\text{s}}(n) + \gamma_{128}^{\text{s}}(n)).
\end{align*}
For fixed $\alpha = 0.95$, we take $f_4(n),f_5(n)$ and $f_6(n)$ from the  fractional Laplacian matrix of a ring graph, a random ring graph and a sphere graph, respectively. 
The details of the original signals are shown in Fig.~\ref{figfm4}(a),(d),(g).
We choose the heat diffusion kernel as the window function $\hat{g}_{\alpha}$.
In Fig.~\ref{figfm4}, we note that the details of the spectrograms by MWGFRFT and FMWGFRFT are the same. 
Moreover, we can observe that the vertex-frequency representations of FMWGFRFT closely match that of MWGFRFT for the 3 signals;
moreover, FMWGFRFT effectively identified the correct frequencies, further validating the correctness and effectiveness of FMWGFRFT algorithm.

\begin{figure*}[!t]
	\centering
    \subfloat[a ring graph]{\includegraphics[width=0.32\textwidth]{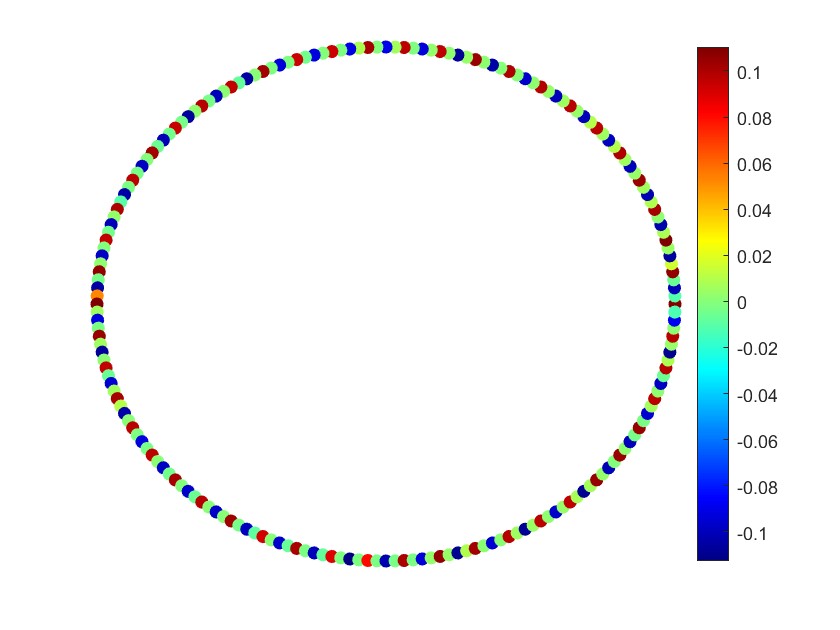}}\hfill
	\subfloat[MWGFRFT of $f_4$]{\includegraphics[width=0.32\textwidth]{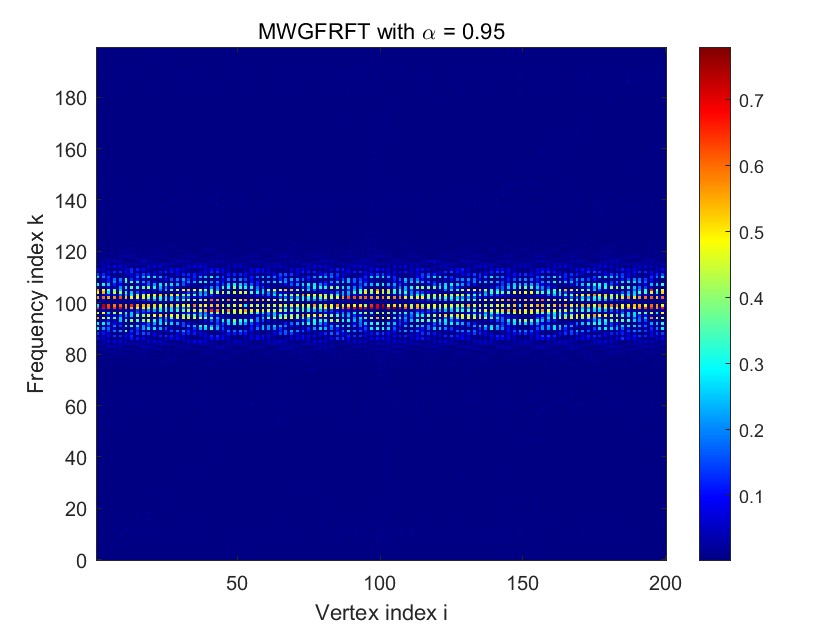}}\hfill
	\subfloat[FMWGFRFT of $f_4$]{\includegraphics[width=0.32\textwidth]{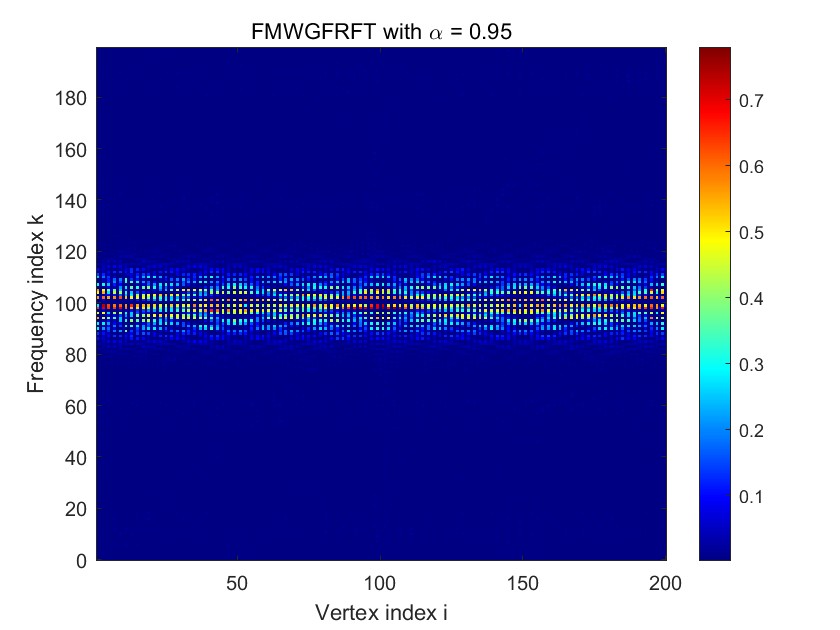}}   \\
	\subfloat[a random ring graph
	]{\includegraphics[width=0.32\textwidth]{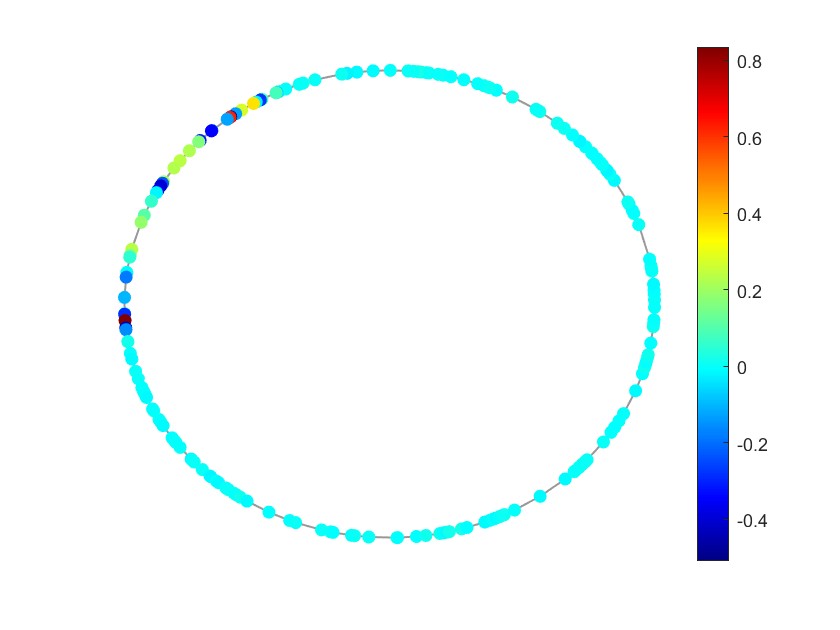}}\hfill     
	\subfloat[MWGFRFT of $f_5$]{\includegraphics[width=0.32\textwidth]{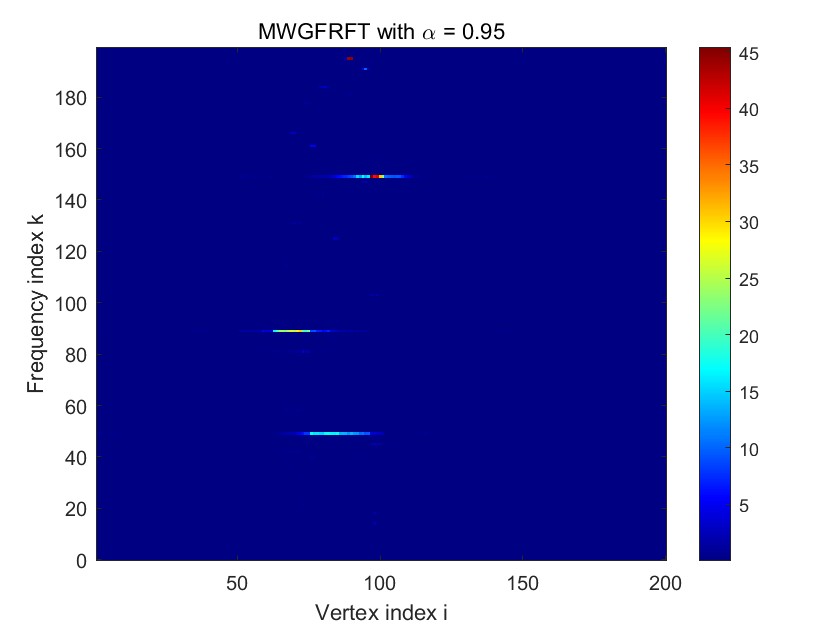}}\hfill
	\subfloat[FMWGFRFT of $f_5$]{\includegraphics[width=0.32\textwidth]{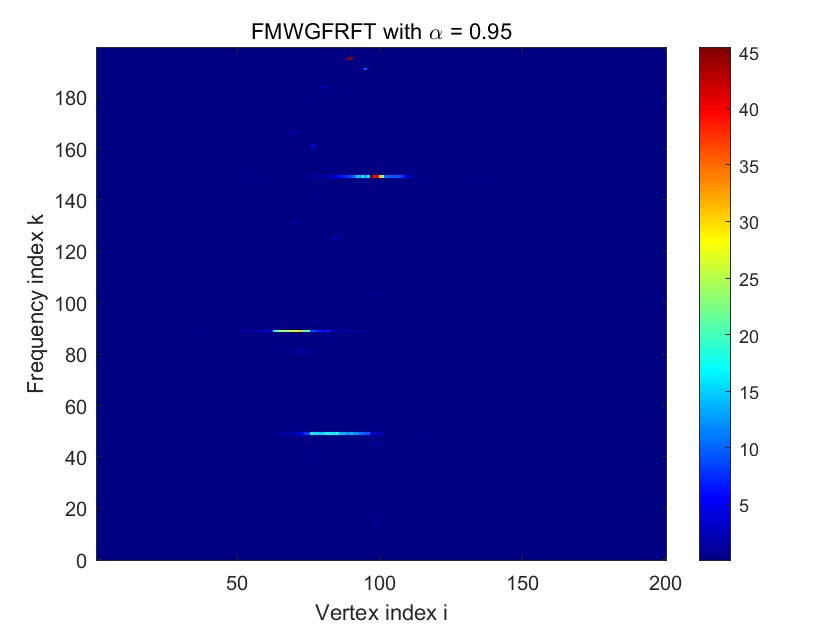}}  \hfill  \\
	\subfloat[a sphere graph
	]{\includegraphics[width=0.32\textwidth]{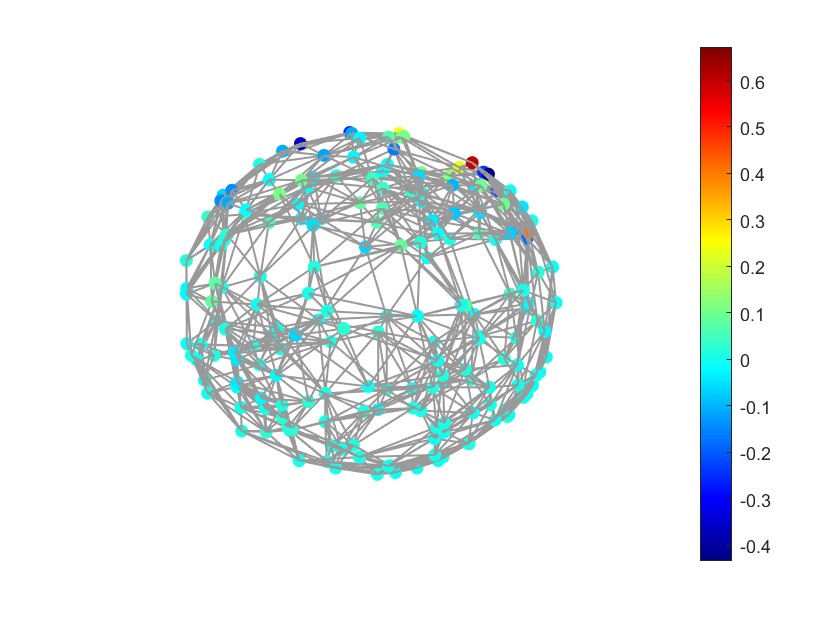}}\hfill
	\subfloat[MWGFRFT of $f_6$]{\includegraphics[width=0.32\textwidth]{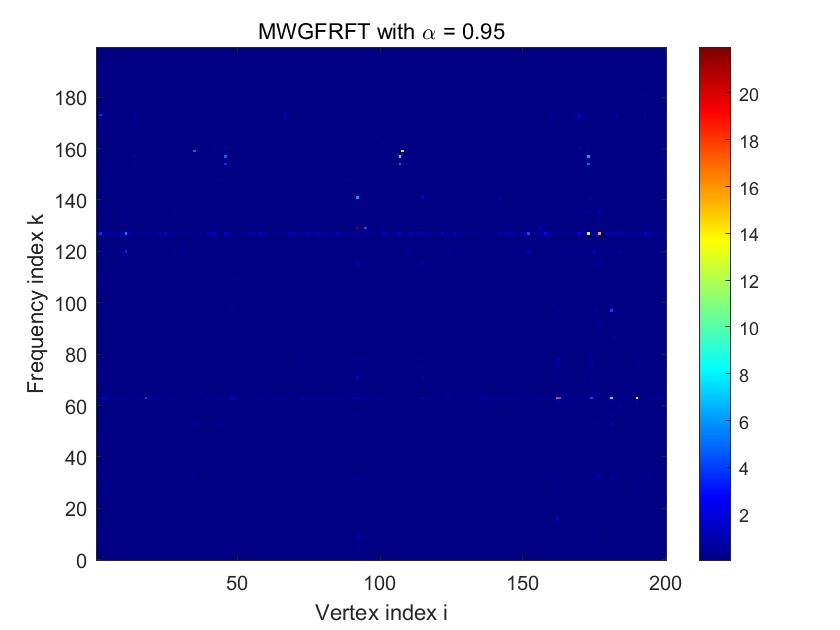}}\hfill
	\subfloat[FMWGFRFT of $f_6$]{\includegraphics[width=0.32\textwidth]{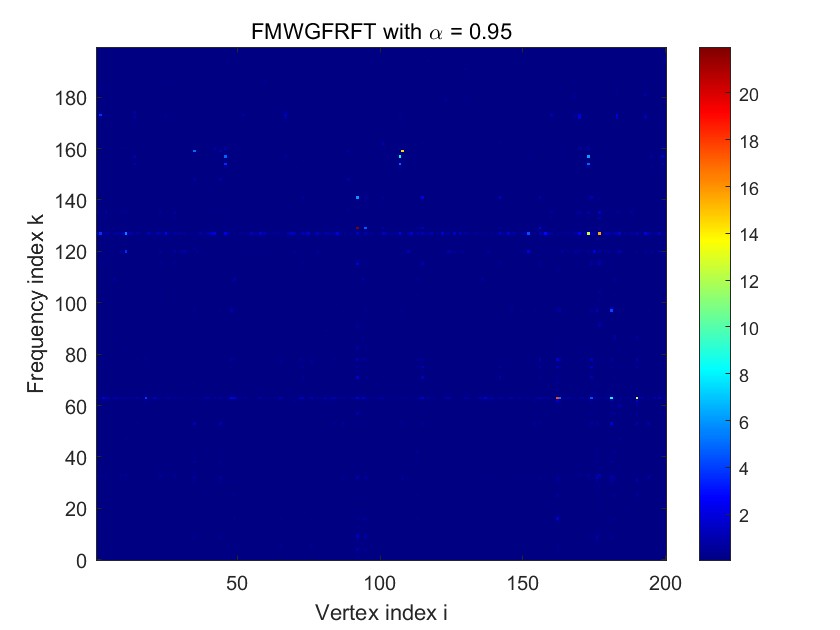}}\\
	\caption{Vertex-frequency representations of signals $f_4$, $f_5$, and $f_6$ by MWGFRFT and FMWGFRFT, with $\alpha = 0.95$.}
	\label{figfm4}
\end{figure*}

In the second experiment, we vary both the fractional order and the noise parameter. 
The normalized mean squared error (NMSE) is computed between the FMWGFRFT coefficients of 3 original signals and their corresponding contaminated signals, as illustrated in Fig.~\ref{figbar5}.
In Fig.~\ref{figbar6}, the WGFRFT results are presented by using a single window function, while all other conditions remain identical to those in Fig.~\ref{figbar5}.

We experiment with 3 graphs of $N = 300$ vertices:
for the sphere graph, we add Poisson noise with parameter $\lambda = 0.2$, $0.3$, and $0.5$ to signal $f_7$;
for the community graph, we add exponential noise with parameter $\lambda = 0.2$, $0.3$, and $0.5$ to signal $f_8$;
for the swiss roll graph, we add Gaussian noise with standard deviation $\sigma = 0.2$, $0.3$, and $0.5$ to signal $f_9$, respectively.
\begin{equation*}
	f_7(n) = \sin\left(110 \pi n/N\right),  \text{for } 1 \leq n \leq 300 
\end{equation*}
\begin{equation*}   
	f_8(n) = 
	\begin{cases} 
		\sin\left(160 \pi n/N\right), & \text{for } 1 \leq n \leq 90 \\
		\sin\left(70 \pi n/N\right), & \text{for } 91 \leq n \leq 170 \\
		\sin\left(200 \pi n/N\right), & \text{for } 171 \leq n \leq 300
	\end{cases}
\end{equation*}
\begin{equation*}
	f_9(n) = \sin\left(\left(30 n + \frac{1}{5} n^2\right) \pi/N\right), \quad 
	\text{for } 1 \leq n \leq 300
\end{equation*} 
For the original signals $f_7$,$f_8$ and $f_9$, the contaminated signals are denoted by $\bar{f}_7$,$\bar{f}_8$ and $\bar{f}_9$.
We choose the standard Gaussian function as the window function by setting $\hat{g}_{\alpha}(\lambda_\ell) = Ce^{-\tau\lambda_\ell^2}$ with $\tau = 0.5$ and choosing $C$ such that $\|\hat{g}_{\alpha}\|_2 = 1$. 
Applying Eq.~\eqref{eq1} to $\hat{g}_{\alpha}$, and then obtain $L=20$ windows.
We then calculate the normalized mean squared error (NMSE) between the transform coefficients of the original signal and the contaminated signal, which is given by: 
$$\mathrm{NMSE}=\frac{\mathrm{MSE}}{\|\mathrm{FWf}\|_2^2}~,$$
where
$\mathrm{MSE} = \frac{1}{N^2} \sum\limits_{i=1}^{N}\sum\limits_{k=0}^{N-1} (FWf - FW\bar{f})^2$,
$N$ is the number of vertices, $FWf$ and $FW\bar{f}$ are the FWGFRFT coefficients of the original and contaminated signals, respectively.

From Fig.~\ref{figbar5} and Fig.~\ref{figbar6}, we observe that, despite differences in graph structure and noise, a higher noise parameter corresponds to a larger value of NMSE.  
Compared to WGFT, MWGFT and WGFRFT, the FMWGFRFT algorithm consistently achieves the lower NMSE, highlighting its robustness and reliability across diverse scenarios.
To minimize the value of NMSE, an appropriate choice of fractional order $\alpha$ can be chosen depending on the specific case.

\begin{figure*}[!t]
	\centering
	\subfloat[NMSE of $FWf_7$ and $FW\bar{f}_7$ 
	]{\includegraphics[width=0.32\textwidth]{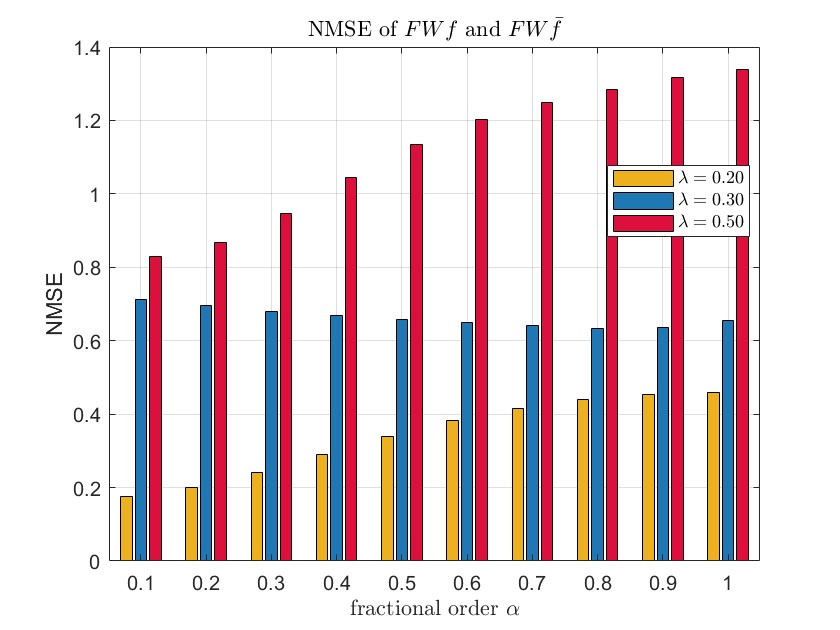}} \hfill
	\subfloat[NMSE of $FWf_8$ and $FW\bar{f}_8$ 
	]{\includegraphics[width=0.32\textwidth]{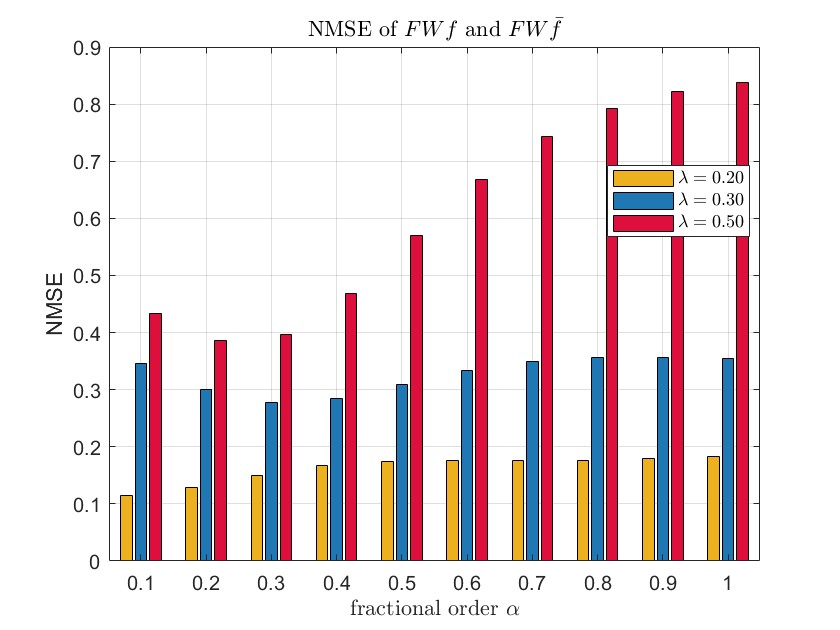}} \hfill
	\subfloat[NMSE of $FWf_9$ and $FW\bar{f}_9$ 
	]{\includegraphics[width=0.32\textwidth]{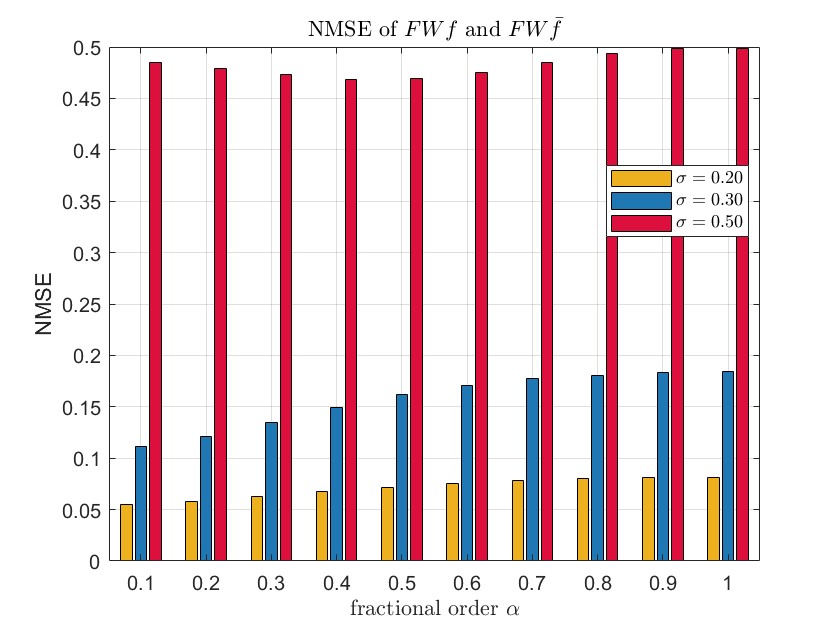}}
	\caption{ NMSE of $FWf$ and $FW\bar{f}$ for $f_7$, $f_8$, and $f_9$ by FMWGFRFT with different fractional order $\alpha$.( It is MWGFT\cite{Zheng2021cc} for $\alpha=1$)}
	\label{figbar5}
\end{figure*}

\begin{figure*}[!t]
	\centering
	\subfloat[NMSE of $Wf_7$ and $W\bar{f}_7$ 
	]{\includegraphics[width=0.32\textwidth]{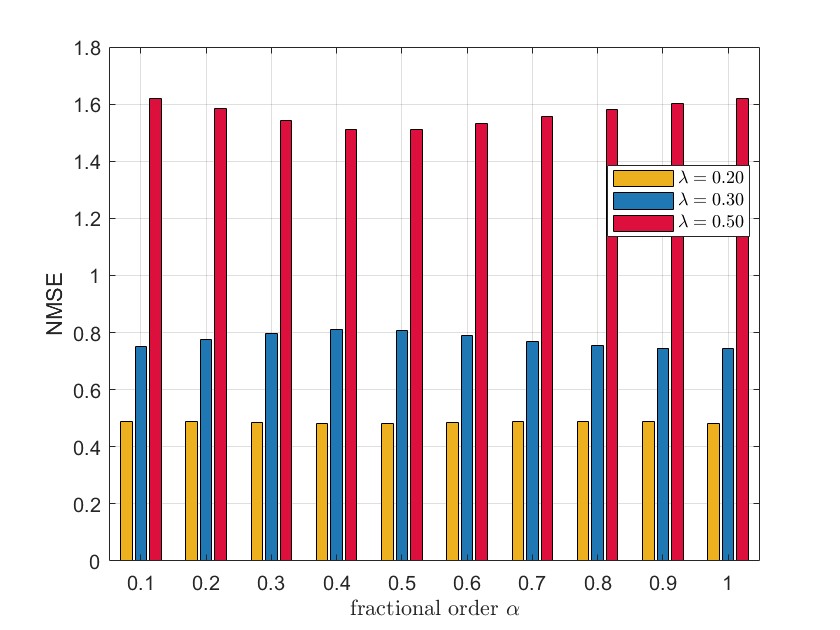}} \hfill
	\subfloat[NMSE of $Wf_8$ and $W\bar{f}_8$ 
	]{\includegraphics[width=0.32\textwidth]{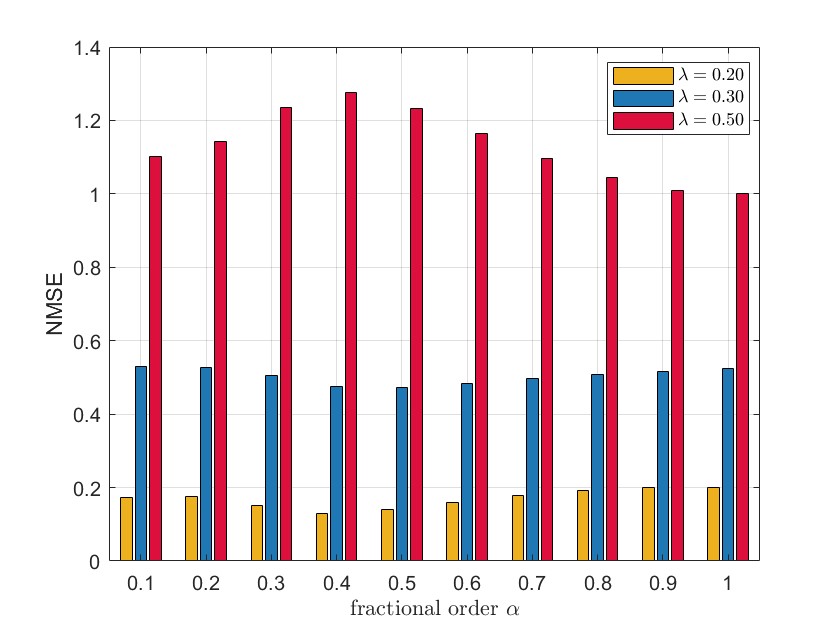}} \hfill
	\subfloat[NMSE of $Wf_9$ and $W\bar{f}_9$ 
	]{\includegraphics[width=0.32\textwidth]{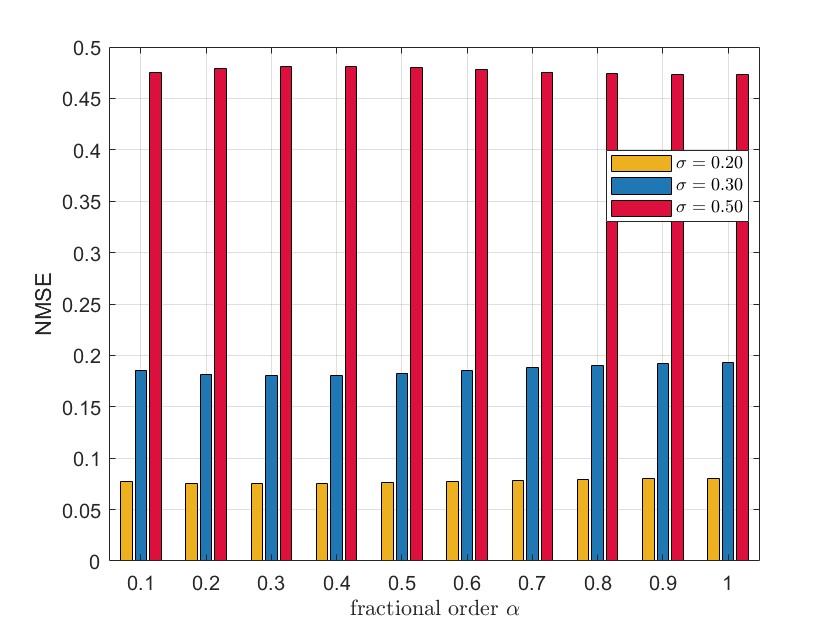}}
	\caption{ NMSE of $Wf$ and $W\bar{f}$ for $f_7$, $f_8$, and $f_9$ by WGFRFT \cite{Yan2021dsp} with different fractional order $\alpha$.( It is WGFT\cite{Shuman2016acha} for $\alpha=1$)}
	\label{figbar6}
\end{figure*}

\subsection{  Examples of SMWGFRFT and FMWGFRFT}

First, we analyze two distinct graph structures: a ring graph with 60 vertices paired with the signal $f_{10}$, and a Swiss roll graph with 60 vertices paired with the signal $f_{11}$.
We choose the heat diffusion kernel as the window function by setting $\hat{g}_{\alpha}(\lambda_\ell) = Ce^{-\tau\lambda_\ell}$ with $\tau = 60$ and choosing $C$ such that $\|\hat{g}_{\alpha}\|_2 = 1$. 
Applying Eq.\eqref{eq1} to $\hat{g}_{\alpha}$, and then obtain $L=10$ windows.
From Fig.~\ref{figssf7}, we know that SMWGFT fails to identify vertex or frequency information, overlooking key vertex-frequency features.
Comparing Fig.~\ref{figssf7}(c) and Fig.~\ref{figssf7}(d), we observe that FMWGFRFT effectively identifies vertex-frequency features. 
In contrast, while SMWGFRFT demonstrates proficiency in capturing vertex-frequency features, it misidentifies certain vertices.
By comparing Fig.~\ref{figssf7}(g) with Fig.~\ref{figssf7}(h), we observe that FMWGFRFT identifies only the prominent vertex-frequency features, disregarding those with less distinct features. 
In contrast, SMWGFRFT captures all vertex-frequency features but introduces some errors in identifying certain points.
Therefore, SMWGFRFT and FMWGFRFT can complement each other in identifying vertex-frequency features and are best utilized in combination.

\begin{figure*}[!t]
	\centering
	\subfloat[$f_{10}$ 
	]{\includegraphics[width=0.23\linewidth]{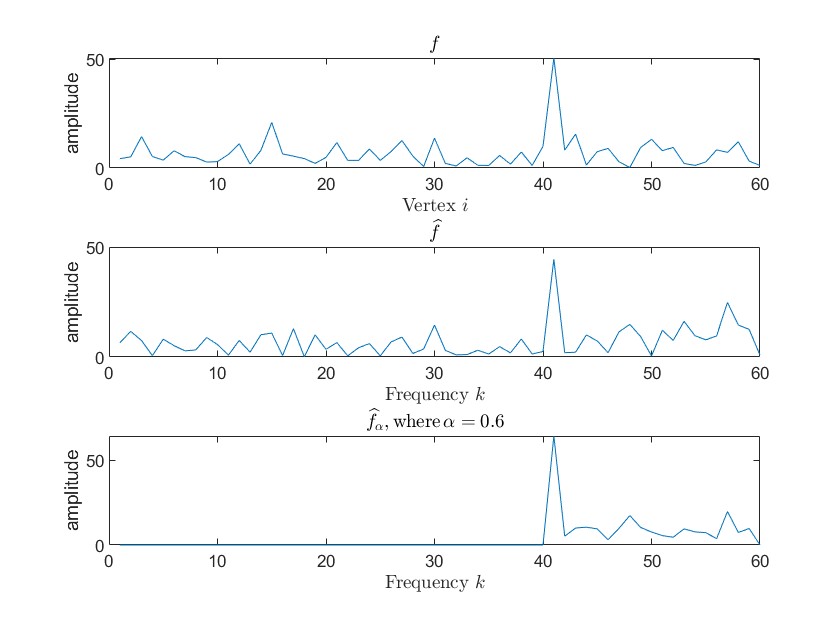}} \hfill
	\subfloat[SMWGFT of $f_{10}$ ]{\includegraphics[width=0.23\linewidth]{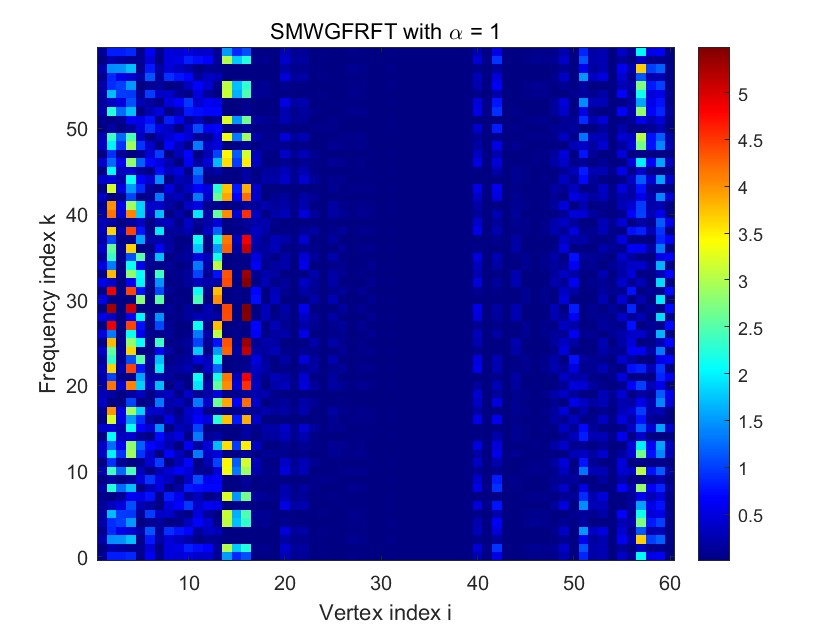}} \hfill
	\subfloat[SMWGFRFT of $f_{10}$
	]{\includegraphics[width=0.23\linewidth]{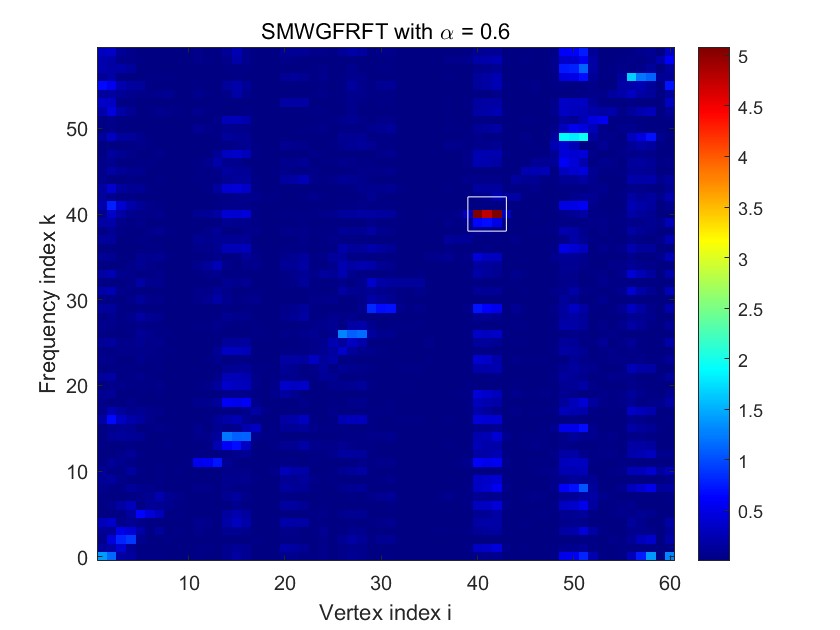}} \hfill
	\subfloat[FMWGFRFT of $f_{10}$
	]{\includegraphics[width=0.23\linewidth]{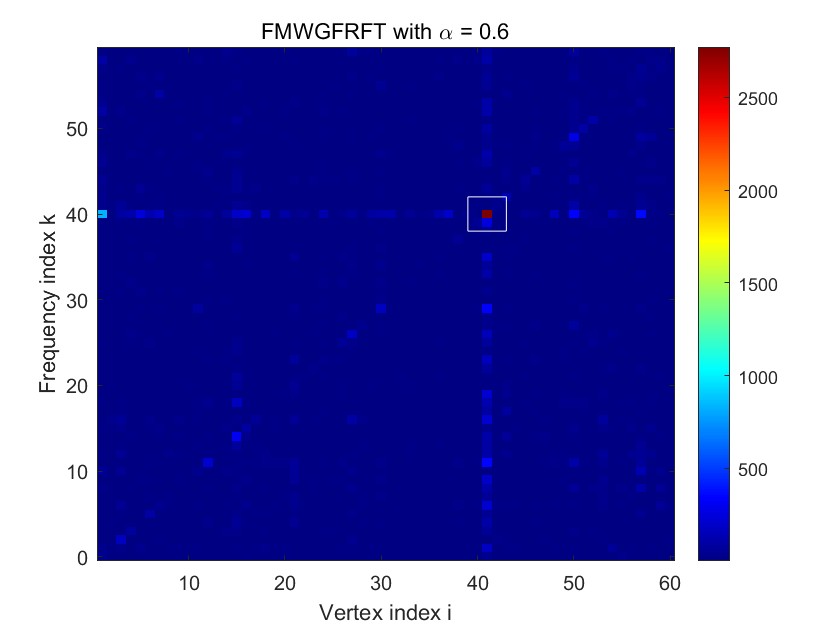}} \\
	\vspace{0.2cm}
	\subfloat[$f_{11}$
	]{\includegraphics[width=0.23\linewidth]{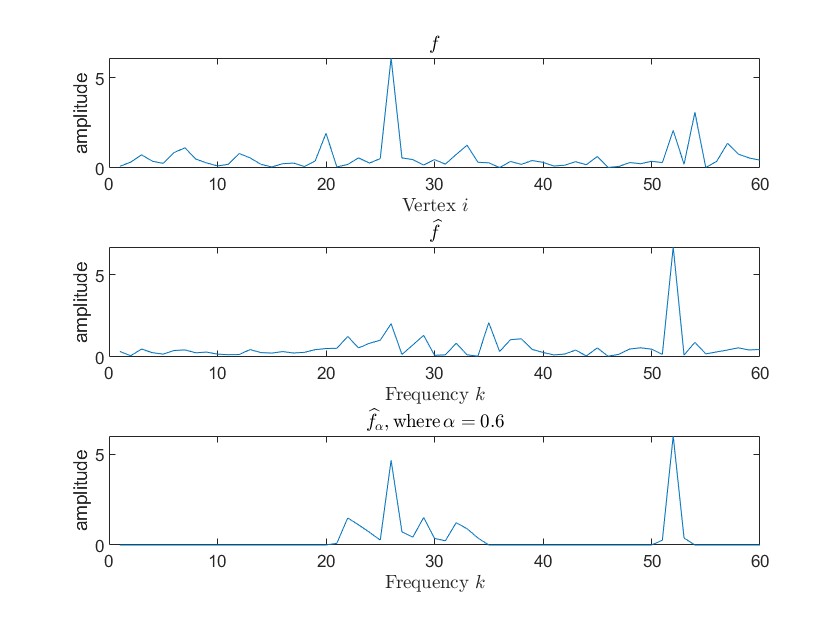}} \hfill
	\subfloat[SMWGFT of $f_{11}$
	]{\includegraphics[width=0.23\linewidth]{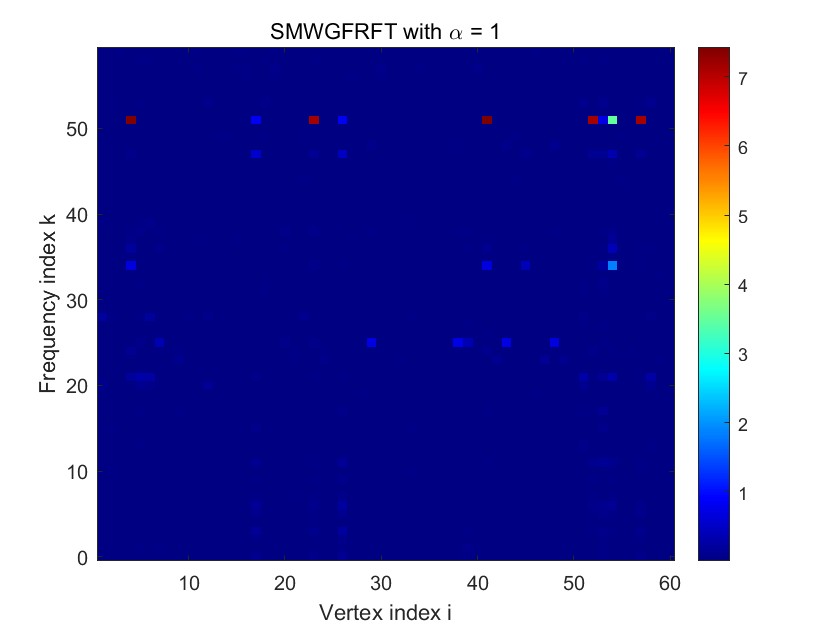}} \hfill
	\subfloat[SMWGFRFT of $f_{11}$
	]{\includegraphics[width=0.23\linewidth]{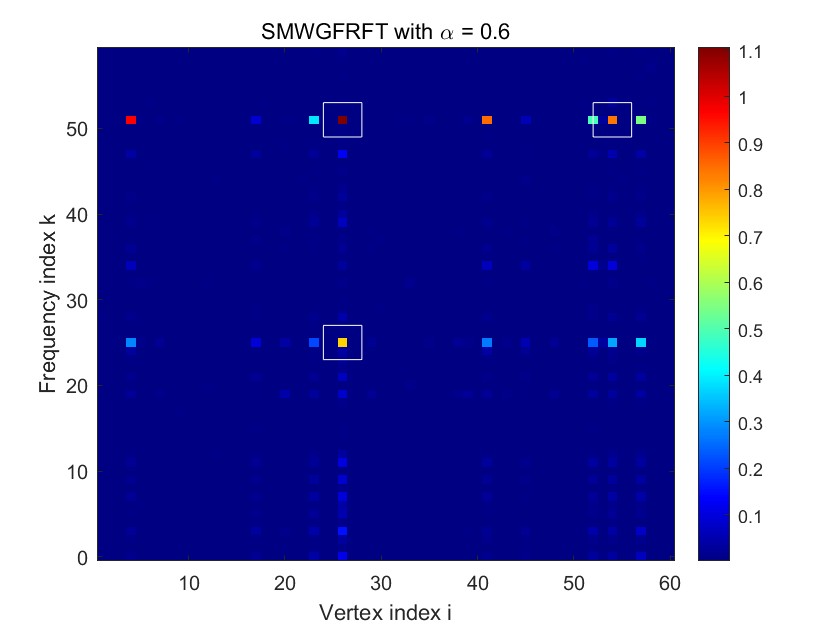}} \hfill
	\subfloat[FMWGFRFT of $f_{11}$
	]{\includegraphics[width=0.23\linewidth]{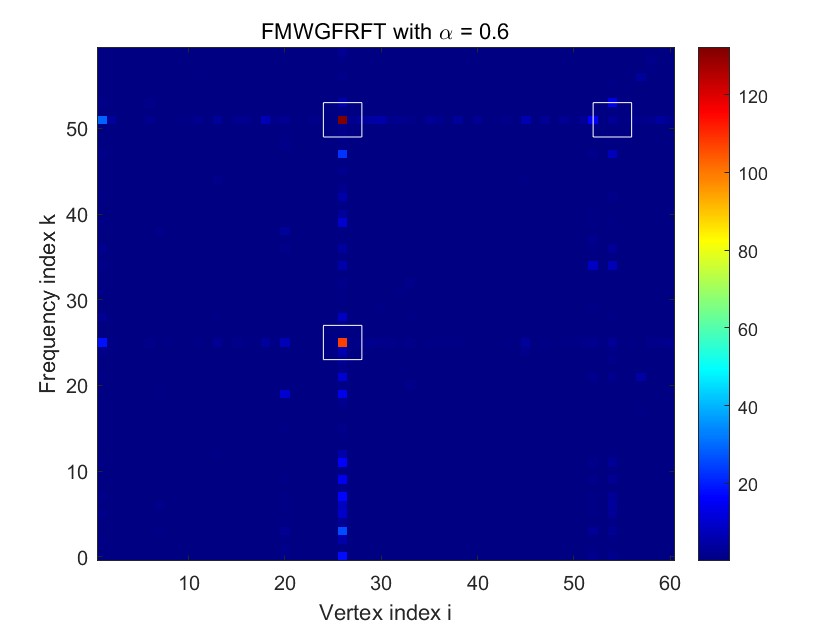}} \\
	\caption{Vertex-frequency representations of signals $f_{10}$ and $f_{11}$ by SMWGFT, SMWGFRFT and FMWGFRFT with fractional order $\alpha = 0.6$.}
	\label{figssf7}
\end{figure*}

Second, we analyze a community graph with 80 vertices and the signal $f_{12}$.
We choose the dual heat diffusion kernel as the window function by setting  $\hat{\tilde{g}}_{\alpha}(\lambda_{l})=\frac{\mu}{e^{-\tau\lambda_{l}}}$ with $\tau = 60$ and $\mu>0$,
choosing $\mu$ such that $\|\hat{\tilde{g}}_{\alpha}\|_2 = 1$. 
Applying Eq.\eqref{eq1} to $\hat{\tilde{g}}_{\alpha}$, and then obtain $L$ windows.
In Fig.~\ref{figs8}, the spectrograms are generated by using dual FMWGFRFT and dual SMWGFRFT with $\alpha=0.6$.
Fig.~\ref{figs8} demonstrates that both dual FMWGFRFT and dual SMWGFRFT are capable of extracting vertex-frequency features.

\begin{figure*}[!t]
	\centering
	\subfloat[ a graph signal $f_{12}$ ]{\includegraphics[width=0.3\textwidth]{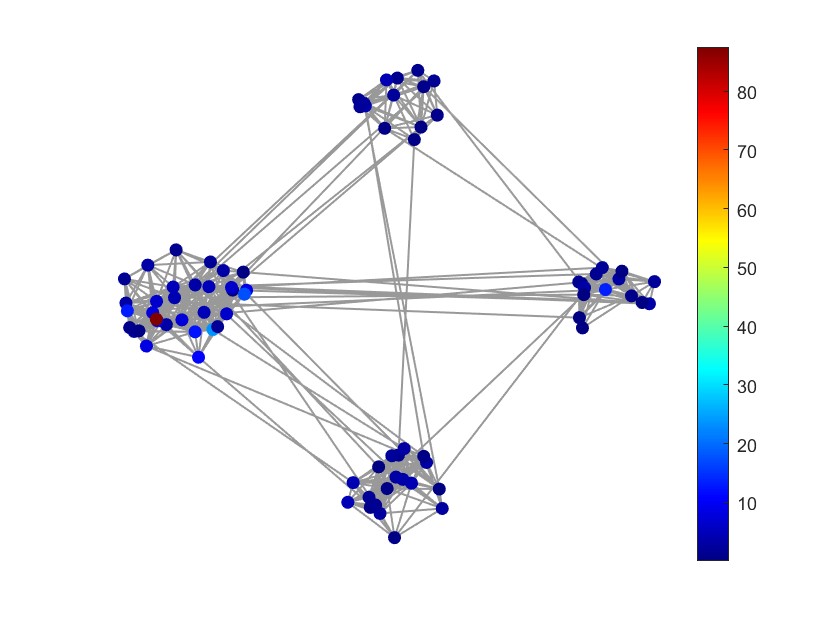}} \hfill
	\subfloat[Dual FMWGFRFT]{\includegraphics[width=0.3\textwidth]{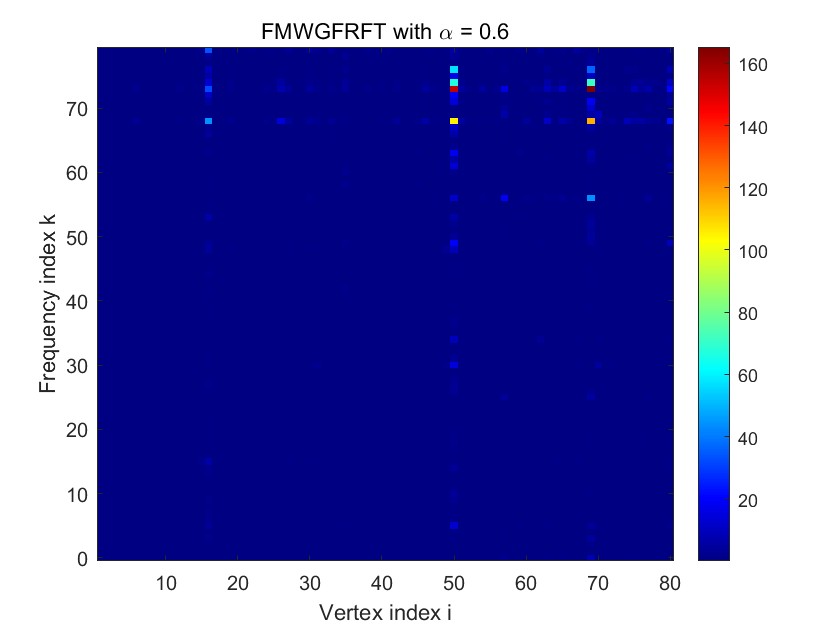}} \hfill
	\subfloat[Dual SMWGFRFT ]{\includegraphics[width=0.3\textwidth]{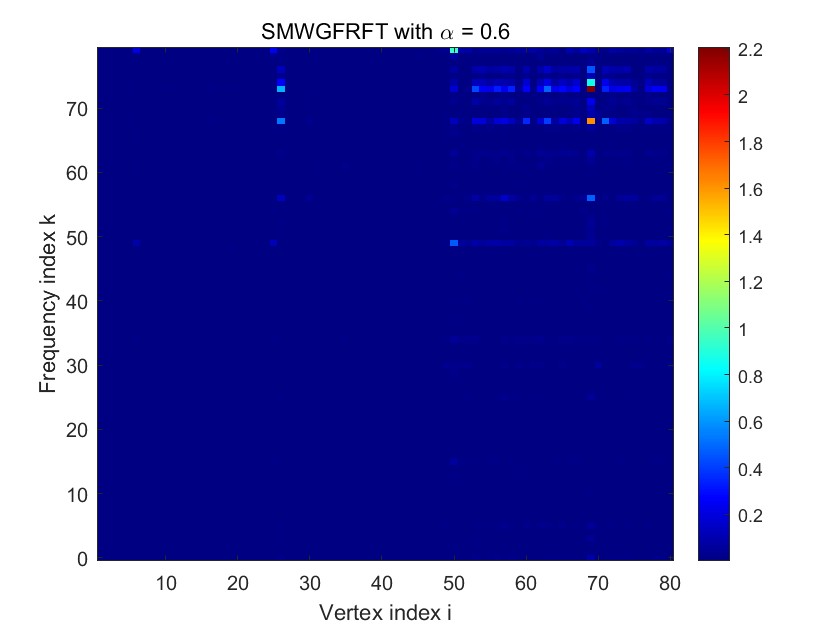}} \hfill
	\caption{The spectrograms generated by dual FMWGFRFT and dual SMWGFRFT with $\alpha=0.6$.}
	\label{figs8}
\end{figure*}

Thirdly, we consider a random ring graph with 50 vertices and the signal $f_{13}$ in Fig.~\ref{figs9}.
The FMWGFRFF window functions is derived from B-spline $\mathbf{N}_{2}$ by Section \ref{section7A}.
The SMWGFRFF window functions is given by the Householder matrix.
By Section \ref{section7B}(ii), the Householder matrix $\mathbf{H}$ is obtained, where
$$\left.\mathbf{v}(n)=\left\{\begin{array}{ll}e^{-0.5(n-1)}&\quad1\leq n\leq10,\\0&\quad n\geq11.\end{array}\right.\right.$$
Fig.~\ref{figs9} demonstrates that both tight FMWGFRFT and Householder tight SMWGFRFT are capable of extracting vertex-frequency features.

\begin{figure*}[!t]
	\centering
	\subfloat[ a graph signal $f_{13}$
	]{\includegraphics[width=0.32\textwidth]{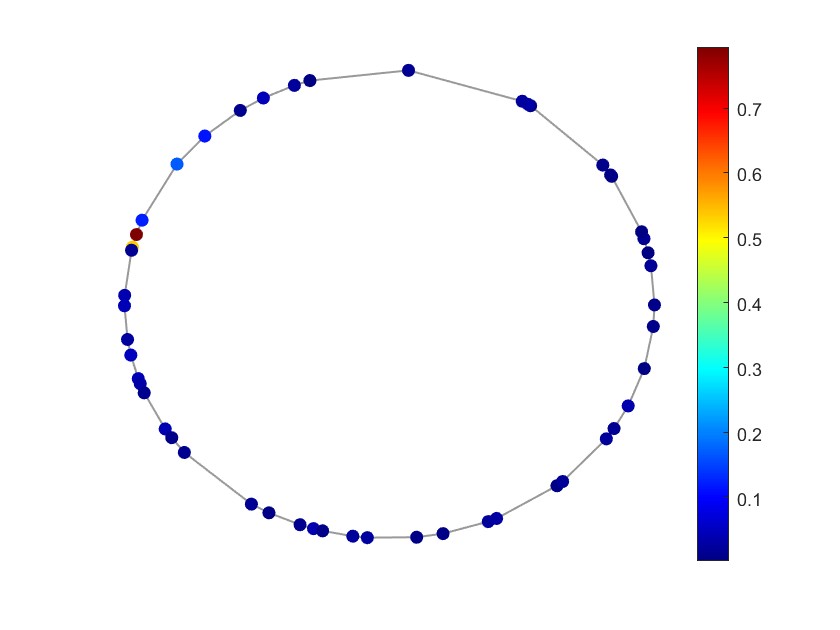}} \hfill
	\subfloat[TFMWGFRFT 
	]{\includegraphics[width=0.32\textwidth]{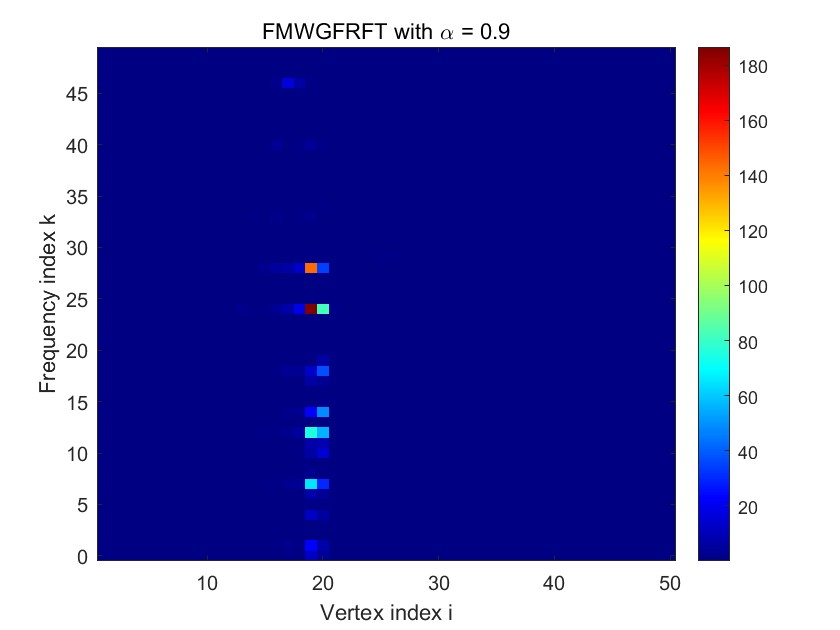}} \hfill
	\subfloat[ Householder TSMWGFRFT
	]{\includegraphics[width=0.32\textwidth]{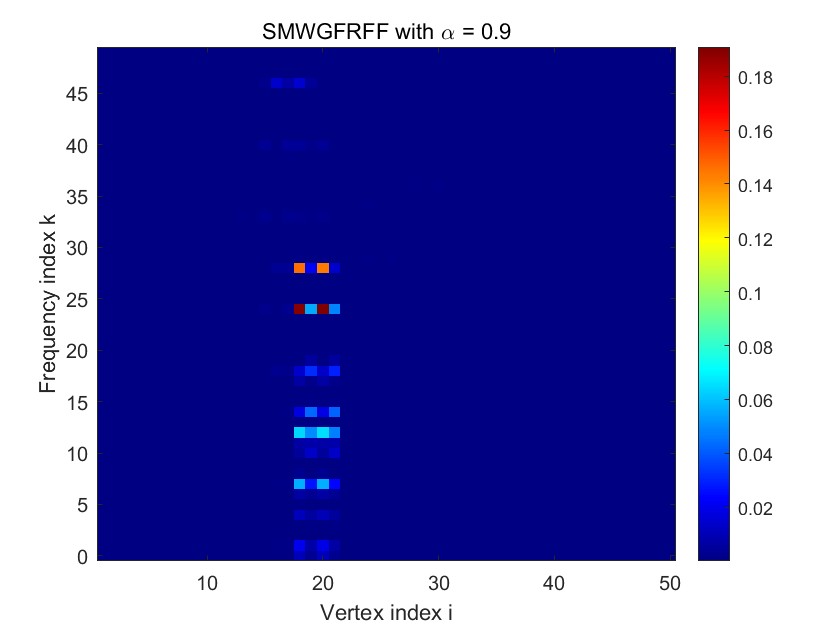}}
	\caption{The spectrograms generated by tight FMWGFRFT and Householder tight SMWGFRFT with $\alpha=0.9$.}
	\label{figs9}
\end{figure*}

\section{ Applications: anomaly detection}\label{section9}

In graph fractional Fourier domain, WGFT \cite{Shuman2016acha}, MWGFT \cite{Zheng2021cc}, WGFRFT \cite{Yan2021dsp}, SMWGFT, SMWGFRFT and FMWGFRFT are applied to detect anomaly signals.
The test involves an anomaly signal $f_{14}$ in Fig.\ref{figsav10}(a), characterized by two anomalous deep red vertices, on a community graph with 20 vertices.
The spectrograms are computed by using WGFT \cite{Shuman2016acha}, MWGFT \cite{Zheng2021cc}, WGFRFT \cite{Yan2021dsp}, SMWGFT, SMWGFRFT and FMWGFRFT.
Taking maximum of $\mathbf{S}$ with respect to the vertices, the spectrogram coefficient threshold is defined as $\delta =  \frac{1}{2}\max(\mathbf{S})$, where $\mathbf{S}$ represents the spectrogram coefficient.
Vertices with maximum spectrogram coefficients exceeding $\delta$ are highlighted in red.
In Fig.\ref{figsav10}(b), the WGFT spectrogram detects all vertices as anomalies.
In Fig.\ref{figsav10}(c), the MWGFT spectrogram identifies all anomalous vertices but incorrectly classifies one additional vertex as anomalous.
In Fig.~\ref{figsav10}(d), the WGFRFT spectrogram fails to detect the anomalous vertices.
In Fig.\ref{figsav10}(e), the SMWGFT spectrogram also fails to detect the anomalous vertices.
In Fig.\ref{figsav10}(f), the SMWGFRFT spectrogram identifies one anomaly vertex but incorrectly classifies too many vertices as anomalous.
In Fig.~\ref{figsav10}(g), the FMWGFRFT spectrogram successfully detects all anomalous vertices.
This example highlights the effectiveness of FMWGFRFT in accurately identifying the vertices associated with the anomaly signals.

\begin{figure*}[!t]
	\centering
	\subfloat[ ]{\includegraphics[width=0.23\linewidth]{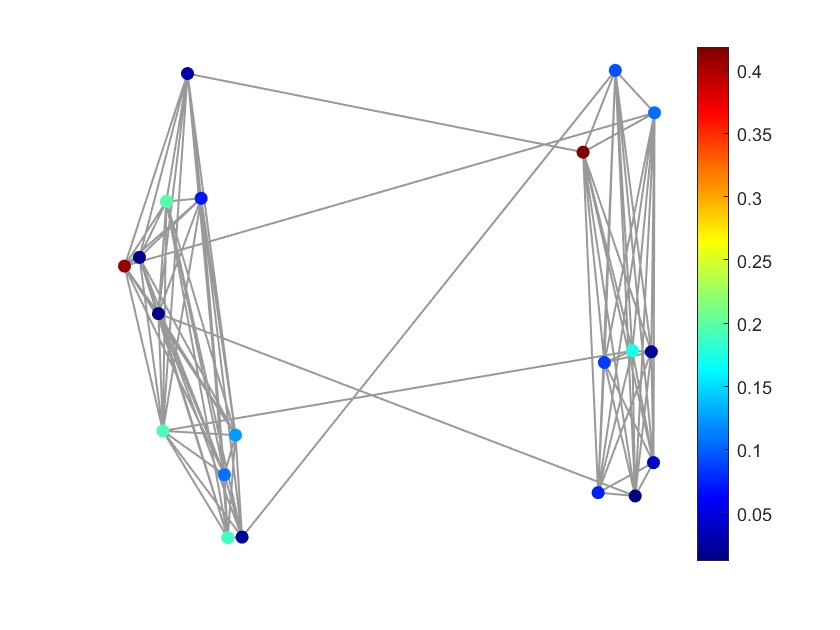}} \hfill
	\subfloat[ ]{\includegraphics[width=0.23\linewidth]{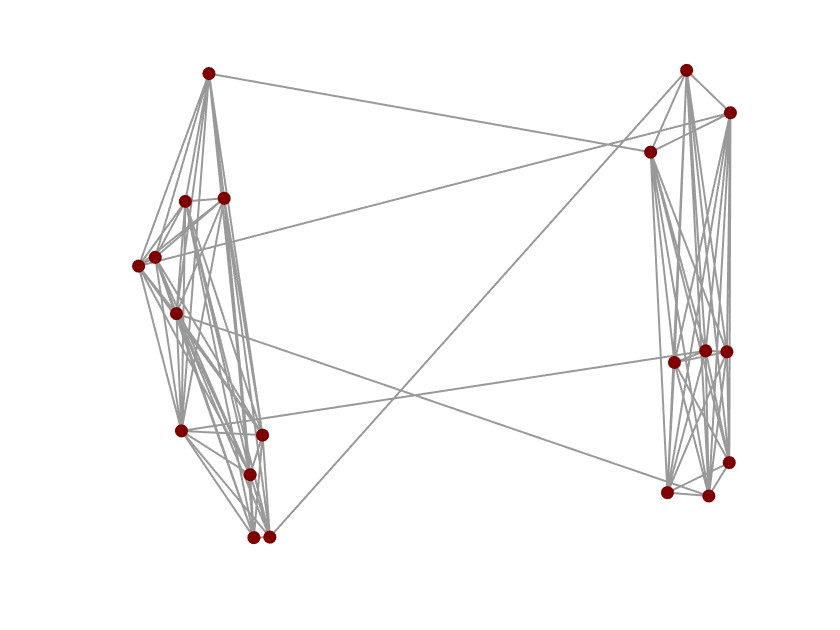}} \hfill
	\subfloat[ ]{\includegraphics[width=0.23\linewidth]{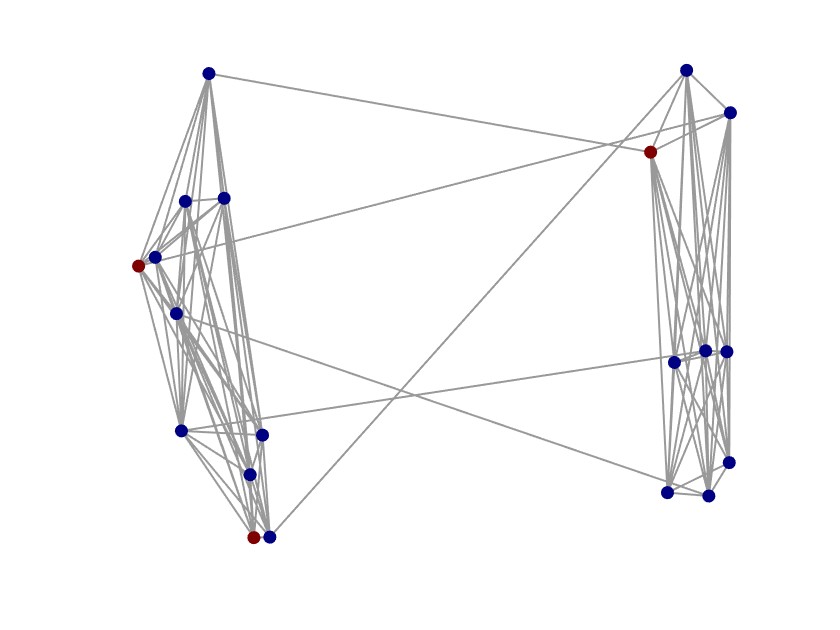}} \hfill
	\subfloat[ ]{\includegraphics[width=0.23\linewidth]{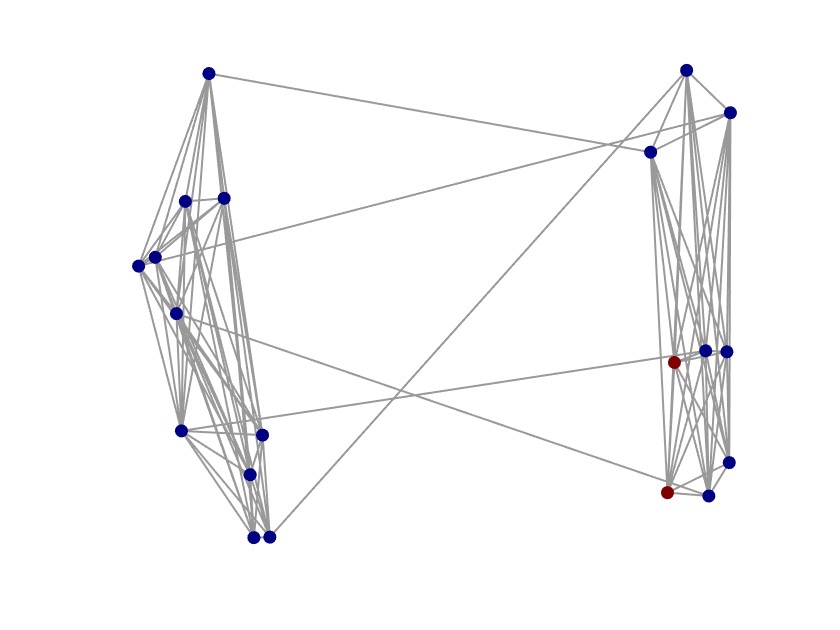}} \\
	\vspace{0.2cm}
	\hspace{0.4cm}\hspace{0.23\linewidth}
	\subfloat[ ]{\includegraphics[width=0.23\linewidth]{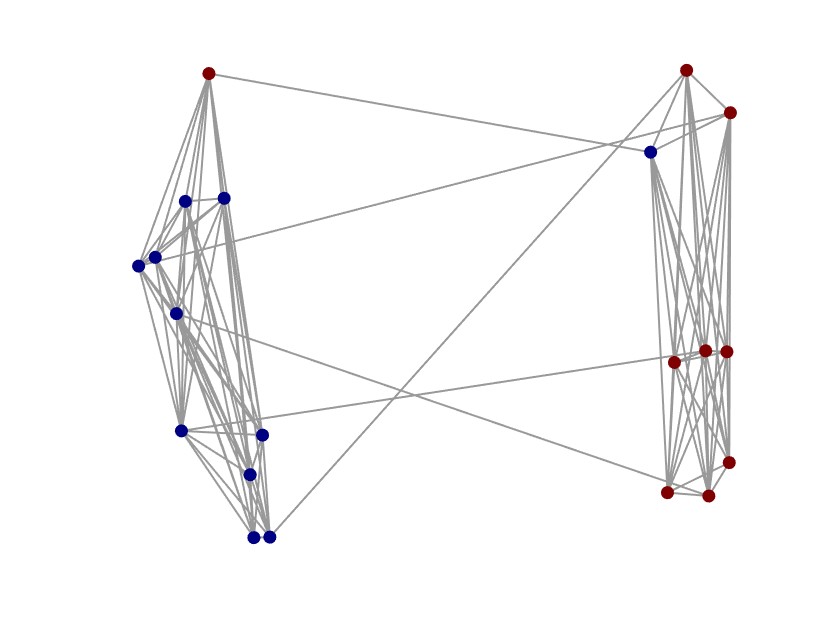}} \hfill
	\subfloat[ ]{\includegraphics[width=0.23\linewidth]{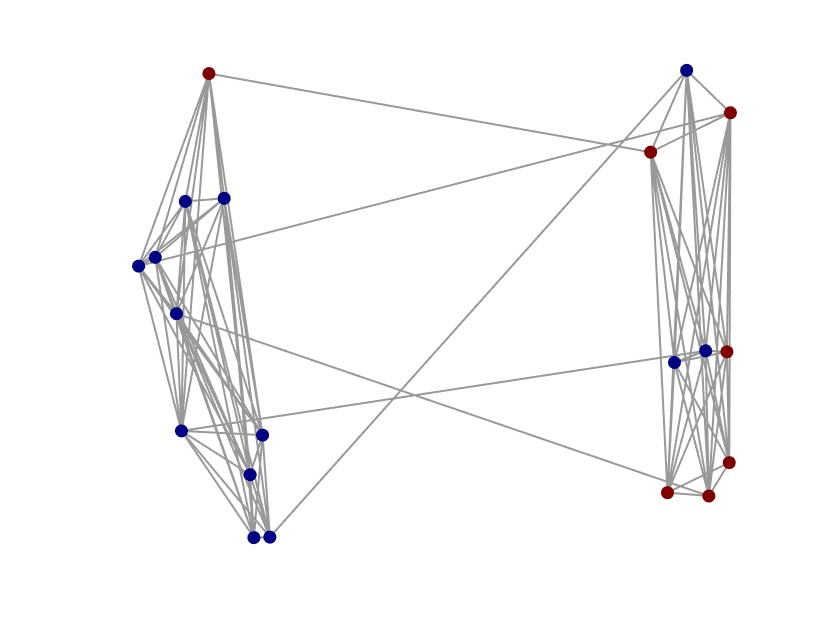}} \hfill
	\subfloat[ ]{\includegraphics[width=0.23\linewidth]{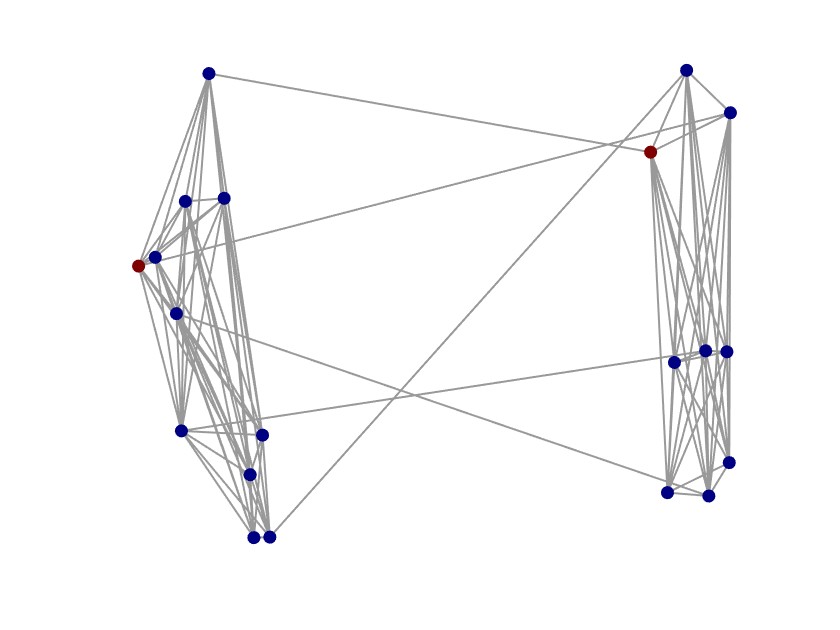}} \hfill \\
	\caption{ 
	(a) An anomaly signal $f_{14}$; 
	(b) Anomaly vertices detected by the WGFT spectrogram on $f_{14}$;
	(c) Anomaly vertices detected by the MWGFT spectrogram on $f_{14}$;
	(d) Anomaly vertices detected by the WGFRFT spectrogram on $f_{14}$;
	(e) Anomaly vertices detected by the SMWGFT spectrogram on $f_{14}$;
	(f) Anomaly vertices detected by the SMWGFRFT spectrogram on $f_{14}$;
	(g) Anomaly vertices detected by the FMWGFRFT spectrogram on $f_{14}$
	(fractional order $\alpha = 0.6$).
	}
	\label{figsav10}
\end{figure*}

\section{Conclusion}\label{section10}

By performing a comparative analysis of vertex-frequency feature extraction among WGFT \cite{Shuman2016acha}, WGFRFT \cite{Yan2021dsp}, MWGFT \cite{Zheng2021cc}, and MWGFRFT, it becomes evident that MWGFRFT offers significant advantages in vertex-frequency representation in the graph fractional Fourier domain for given window functions.
The performance of MWGFRFT and FMWGFRFT is compared with regard to vertex-frequency representation, normalized mean squared error, and time cost. The experimental results demonstrate the robustness and effectiveness of the FMWGFRFF algorithm.
Experiments with SMWGFRFT and FMWGFRFT on different graphs demonstrate that SMWGFRFT and FMWGFRFT effectively extract vertex-frequency features.
Furthermore, SMWGFRFT and FMWGFRFT are mutually complementary in identifying vertex-frequency features and are most effective when utilized in tandem.
Additionally, dual and tight frame window functions have demonstrated effectiveness in extracting vertex-frequency features.
In the graph fractional domain, FMWGFRFT demonstrates effectiveness in accurately detecting anomaly signals.

\appendices

\section{Proof of Theorem \ref{theorem1}}\label{atheorem1}
For any $\mathbf{f}\in\mathbb{C}^{N}$,
\begin{align}  \label{eqthp1}
	&\sum_{l=1}^{L}\sum_{i=1}^{N}\sum_{k=0}^{N-1}|\langle\mathbf{f},\mathbf{g}_{i,k}^{(l\alpha)}\rangle|^{2}\nonumber\\
	&=N^\alpha\sum_{l=1}^{L}\sum_{i=1}^{N}\sum_{k=0}^{N-1}|\langle\mathbf{f}\circ(T_i^\alpha \mathbf{g}_l)^{*},\gamma_{k}\rangle|^{2}\nonumber\\
	&=N^\alpha\sum_{l=1}^{L}\sum_{i=1}^{N}\sum_{n=1}^{N}|\mathbf{f}(n)|^{2}|(T_{n}^\alpha\mathbf{g}_{l})(i)|^{2}\nonumber\\
	&=N^\alpha\sum_{n=1}^{N}|\mathbf{f}(n)|^{2}\sum_{l=1}^{L}\|T_{n}^\alpha\mathbf{g}_{l}\|_{2}^{2},
\end{align}
where \eqref{eqthp1} follows from Parseval relation, the symmetry of $\mathcal{L}_\alpha$ and the definition of
operator $T_{i}^\alpha$. 
In addition, if $\sum\limits_{l=1}^{L}\left|\hat{\mathbf{g}}_{l\alpha}(0)\right|^{2}\neq0$, we have
$$\begin{aligned}\sum_{l=1}^{L}\|T_{n}^\alpha\mathbf{g}_{l}\|_{2}^{2}
	&=N^\alpha\sum_{p=0}^{N-1}\sum_{l=1}^{L}|\hat{\mathbf{g}}_{l\alpha}(r_p)|^{2}|\gamma_{p}(n)|^{2}\\
	&\geq\sum_{l=1}^{L}\left|\hat{\mathbf{g}}_{l\alpha}(0)\right|^{2}>0,
\end{aligned}$$
taking the minimum and maximum of 
$\sum_{l=1}^{L}\|T_{n}^\alpha\mathbf{g}_{l}\|_{2}^{2}$, 
then we have the lower frame bound $A>0.$ 
Thus, $\mathcal{ G}_L^{w}$ is a frame with lower and upper frame bounds defined in \eqref{lfb1} and \eqref{lfb2}.

\section{Proof of Corollary \ref{cor32}}  \label{acor32}
From \eqref{eqsfc1}, the frame operator of $\mathcal{G}_{l\alpha}^{w}$ can be written as the diagonal matrix
$S=\mathbf{D}_{c}=diag(\mathbf{c}),$
with $c_n=N^\alpha\sum_{l=1}^{L}\|T_{n}^\alpha\mathbf{g}_{l}\|_{2}^{2}=C$ for $n= 1, 2,\ldots,N$.
Note that the optimal lower and upper frame bounds of a frame are the smallest and largest eigenvalues of the frame operator, respectively \cite{Casazza2012Springer},
the frame bounds of $\mathcal{G}_{l\alpha}^w$ are given by the smallest and largest entries in $\mathbf{c}$. 
Therefore, $\mathcal{G}_{l\alpha}^w$ is a tight frame if and only if $\mathbf{c}$ is a constant vector; that is, there exists a constant $C$ such that $c_n = C$ for $n = 1, 2, \dots, N$.

\section{Proof of Corollary \ref{cor33}}  \label{acor33}
In Appendix \ref{atheorem1}, for any $\mathbf{f}\in\mathbb{C}^N$,
\begin{align}\label{cor33pf}
	&\sum_{l=1}^{L}\sum_{i=1}^{N}\sum_{k=0}^{N-1}|\langle\mathbf{f},\mathbf{g}_{i,k}^{(l\alpha)}\rangle|^{2}=N^\alpha\sum_{n=1}^{N}|\mathbf{f}(n)|^{2}\sum_{l=1}^{L}\|T_{n}^\alpha\mathbf{g}_{l}\|_{2}^{2}\nonumber\\
	&=N^{2\alpha}\sum_{n=1}^{N}|\mathbf{f}(n)|^{2}\sum_{p=0}^{N-1}\sum_{l=1}^{L}|\hat{\mathbf{g}}_{l\alpha}(r_p)|^{2}|\gamma_{p}(n)|^{2}.
\end{align}
Since the eigen-matrix $\gamma$ of $\mathcal{L}_\alpha$ is orthogonal, 
we have $\sum_{p=0}^{N-1}|\gamma_p(n)|^2=1$, for $n=1,2,...,N$.
If $\sum_{l=1}^{L}\left|\hat{\mathbf{g}}_{l\alpha}(r_{p})\right|^{2}=C$
for $p= 0, 1, \ldots , N- 1$, 
from \eqref{cor33pf}, we have
$$\sum_{l=1}^{L}\sum_{i=1}^{N}\sum_{k=0}^{N-1}|\langle\mathbf{f},\mathbf{g}_{i,k}^{(l\alpha)}\rangle|^{2}=N^{2\alpha}C\|\mathbf{f}\|_{2}^{2}.$$
Thus, $\mathcal{G}_{l\alpha}^{w}$ is a tight frame with frame bounds $A= B=N^{2\alpha}C.$ 

\section{Proof of Proposition \ref{ledsvf1}}\label{aledsvf1}
The proofs of FMWGFRFT, MWGFRFT, MWGFT\cite{Zheng2021cc}, WGFRFT\cite{Yan2021dsp} and WGFT\cite{Shuman2016acha} are analogous. 
For the sake of convenience, we will only provide the proof for FMWGFRFT.
Note that the FMWGFRFT matrix
$$FWf= N^\alpha\sum_{l=1}^L \gamma*\left(
(\Psi^{l\alpha})^*\circ \left((F\circ\gamma^H)*\gamma\right)\right)^{.T}.$$
As we all know, when $\alpha$ tends to 0, the graph fractional Fourier basis matrix $\gamma$ tends to the identity matrix, i.e., $\lim\limits_{{\alpha \to 0}} \gamma = I$.
Then we have
\begin{align}
	&\lim\limits_{{\alpha \to 0}}FWf   
	=\lim\limits_{{\alpha \to 0}}   \begin{pmatrix}
		\sum\limits_{l=1}^L\mathbf{f}(1)\hat{\mathbf{g}}_{l\alpha}(r_0)&0\\
		~~~~~~~~~~~~~~~~~~\ddots\\
		0&\sum\limits_{l=1}^L\mathbf{f}(N)\hat{\mathbf{g}}_{l\alpha}(r_{N-1})
	\end{pmatrix}      \nonumber  
\end{align}
We know that as $\alpha$ tends to 0, the vertex-frequency representation of FMWGFRFT is clustered on the main diagonal, i.e., all values outside the main diagonal equal to 0.

\section{Proof of Theorem \ref{theorem2}}  \label{atheorem2}
Suppose that
$\sum_{l=1}^{L}\mathbf{\hat{g}}_{l\alpha}^{*}(r_{p})\mathbf{\hat{\tilde{g}}}_{l\alpha}(r_{p})=\mu$ for $p=0,1,\ldots,N-1$. 
Then we have
$$\begin{aligned}
	&\sum_{l=1}^{L}\sum_{i=1}^{N}\sum_{k=0}^{N-1}\langle \mathbf{f},{\mathbf{g}}_{i,k}^{(l\alpha)}\rangle\tilde{\mathbf{g}}_{i,k}^{(l\alpha)}\\
	&=N^{2\alpha}\sum_{m=1}^{N}\mathbf{f}(m)\sum_{l=1}^{L}\sum_{p=0}^{N-1}\sum_{p^{\prime}=0}^{N-1}\mathbf{\hat{g}}_{l\alpha}^{*}(r_{p})\mathbf{\hat{\tilde{g}}}_{l\alpha}(r_{p^{\prime}})\gamma_{p}^{*}(m)\gamma_{p^{\prime}}(n)\cdot \\
	&\sum_{i=1}^{N}\gamma_{p}(i)\gamma_{p^{\prime}}^{*}(i)\sum_{k=0}^{N-1}\gamma_{k}^{*}(m)\gamma_{k}(n)\\
\end{aligned}$$
$$\begin{aligned}
	&=N^{2\alpha}\sum_{m=1}^{N}\mathbf{f}(m)\sum_{l=1}^{L}\sum_{p=0}^{N-1}\sum_{p^{\prime}=0}^{N-1}\mathbf{\hat{g}}_{l\alpha}^{*}(r_{p})\mathbf{\hat{\tilde{g}}}_{l\alpha}(r_{p^{\prime}})\gamma_{p}^{*}(m)\gamma_{p^{\prime}}(n)\cdot \\
	&\delta_{pp^{\prime}}\delta_{mn} \\
	&=N^{2\alpha}\mu\mathbf{f}(n). 
\end{aligned}$$
Thus, $\tilde{\mathcal{G}}_{l\alpha}^{w}$ is a dual of ${\mathcal{G}}_{l\alpha}^{w}$ with $C=\frac{1}{ N^{2\alpha}\mu}$ in \eqref{eq51}.

\section{Proof of Corollary \ref{cor42}}  \label{acor42}
If $\sum_{l=1}^{L}\mathbf{\hat{g}}_{l\alpha}^{*}(r_{p})\mathbf{\hat{\tilde{g}}}_{l\alpha}(r_{p})=\mu$ for $p=0,1,\ldots,N-1$,
by the Cauchy-Schwartz inequality, we have
$$\mu^{2}
=|\sum_{l=1}^{L}\mathbf{\hat{g}}_{l\alpha}^{*}(r_{p})\mathbf{\hat{\tilde{g}}}_{l\alpha}(r_{p})|^{2}
\leq\sum_{l=1}^{L}|\mathbf{\hat{g}}_{l\alpha}^{*}(r_{p})|^{2}\sum_{l=1}^{L}|\mathbf{\hat{\tilde{g}}}_{l\alpha}(r_{p})|^{2}.$$
Then $\sum_{l=1}^{L}|\mathbf{\hat{g}}_{l\alpha}^{*}(r_{p})|^{2} \neq 0$
and $\sum_{l=1}^{L}|\mathbf{\hat{\tilde{g}}}_{l\alpha}(r_{p})|^{2}\neq 0$ for $p=0,1,\ldots,N-1$. 
By Theorem \ref{theorem1},we have ${\mathcal{\tilde{G}}}_{l\alpha}^{w}$ is also a multi-windowed graph fractional Fourier frame.

\section{Proof of Corollary \ref{cor43}}  \label{acor43}
Suppose that $S$ is the frame operator of ${\mathcal{G}}_{l\alpha}^{w}$. 
In the proof of Corollary \ref{cor32}, we have that
$S=\mathbf{D}_{c}=\mathrm{diag}(\mathbf{c}).$
According to the definitions of frame and operator, the canonical frame of ${\mathcal{\tilde{G}}}_{l\alpha}^{w}$ is then given by
$$S^{-1}\mathbf{g}_{i,k}^{(l\alpha)}
=\mathbf{D}_{c}^{-1}\mathbf{g}_{i,k}^{(l\alpha)}
=\mathbf{D}_{d}\mathbf{g}_{i,k}^{(l\alpha)}
=\mathbf{d}\circ\mathbf{g}_{i,k}^{(l\alpha)}.$$

\section{Proof of Theorem \ref{theorem3}}   \label{atheorem3}
We have $$\mathbf{g}_{i,l}=\mathbf{D}_{i}S^\alpha\mathbf{g}_{l},$$ 
where $\mathbf{D}_i=diag(\gamma_{1,i},\gamma_{2,i},\cdots,\gamma_{N,i})$.
Let
$T_l=(\mathbf{D}_1S^\alpha\mathbf{g}_l,\mathbf{D}_2S^\alpha\mathbf{g}_l\cdots,\mathbf{D}_NS^\alpha\mathbf{g}_l)$,
the corresponding synthesis operator of $\mathcal{G}_{l\alpha}^{s}$ can be expressed as
$T=(T_1,T_2\cdots,T_L)$.
We know that $\mathcal{G}_{l\alpha}^{s}$ forms a frame if and only if its frame operator is positive definite. 
In fact, the frame operator of $\mathcal{G}_{l\alpha}^{s}$ can be expressed as
$$	S
	=\sum_{l=1}^L\sum_{i=1}^N[S^\alpha \mathbf{g}_l\mathbf{g}_l^*S^{\alpha*}]\circ(\gamma_i\gamma_i^*)
	=\sum_{l=1}^{M}\bigl[\bigl(S^\alpha \mathbf{g}_l\bigr)\bigl(S^\alpha \mathbf{g}_l\bigr)^{*}\bigr]\circ\mathbf{I}_{N}. $$
Here $\mathbf{I}_N$ is an identity matrix.
Let $\mathbf{G}= \sum _{i= 1}^{L}\mathbf{g} _{i}\mathbf{g} _{i}^{* }$,
we have
$$\sum_{l=1}^L(S^\alpha \mathbf{g}_l\bigr)\bigl(S^\alpha \mathbf{g}_l\bigr)^{*}
=S^\alpha\mathbf{G}S^{\alpha*}.$$

Let $\tilde{\mathbf{s}}_k$ is the $k$th row of matrix $S^\alpha$, 
the $k$th diagonal entry of $S^\alpha\mathbf{G}S^{\alpha*}$ can be written as $c_k:=\tilde{\mathbf{s}}_k\mathbf{G}\tilde{\mathbf{s}}_k^*$.
Therefore, $\mathcal{G}_{l\alpha}^{s}$ is a frame if and only if $c_k>0$ for $k=1,\ldots,N$.
The frame bounds of $\mathcal{G}_{l\alpha}^{s}$ are given by \eqref{eqth32} which are the smallest and largest elements of $\mathbf{c}$. 

\section{Proof of Corollary \ref{cor53}}  \label{acor53}
By Theorem \ref{theorem3}, the frame bounds of $\mathcal{G}_{l\alpha}^{s}$ are given by the smallest and largest elements in $\mathbf{c}$, with $c_k=\tilde{\mathbf{s}}_k\mathbf{G}\tilde{\mathbf{s}}_k^*$.
Hence, $\mathcal{G}_{l\alpha}^{s}$ is a tight frame is then equivalent to the condition that $\tilde{\mathbf{s}}_k\mathbf{G}\tilde{\mathbf{s}}_k^*=C$ for
$k= 1, \ldots, N$.

\balance

\bibliographystyle{IEEEtran}
\bibliography{main}

\end{document}